\documentclass[11pt, letterpaper]{article}

\usepackage{graphicx}
\usepackage{epstopdf}
\usepackage{amsmath,amssymb,amsfonts}
\usepackage{setspace}
\usepackage{epstopdf}
\usepackage{color}
\usepackage[letterpaper,left=1in,right=1in,top=1in,bottom=1in]{geometry}
\usepackage{multirow}
\usepackage{longtable}
\usepackage{rotating}
\usepackage{mathrsfs}
\usepackage{algorithm}
\usepackage{algorithmicx}
\usepackage{algpseudocode}
\usepackage{booktabs}
\usepackage{ntheorem}
\usepackage{arydshln}

\newtheorem{remark}{Remark}

\newtheorem*{proof}{Proof}
\newtheorem{lemma}{Lemma}
\newtheorem{proposition}{Proposition}

\newtheorem{theorem}{Theorem}

\newtheorem{assumption}{Assumption}

\title{\Large \bf Approximating arrival costs in distributed moving horizon estimation: A recursive method}

\author{
\centerline{\normalsize Xiaojie Li$^{a}$, Xunyuan Yin$^{a,b,}$\thanks{Corresponding author: Xunyuan Yin. Tel: (+65)63168746. Email: xunyuan.yin@ntu.edu.sg}
}
\vspace{5mm}\\
\centerline{\small $^{a}$School of Chemistry, Chemical Engineering and Biotechnology, Nanyang Technological University,}\\
\centerline{\small 62 Nanyang Drive, 637459, Singapore}\\
\centerline{\small $^{b}$Environmental Process Modelling Centre, Nanyang Environment and Water Research Institute (NEWRI),}\\
\centerline{\small Nanyang Technological University, 1 CleanTech Loop, 637141, Singapore}}
\allowdisplaybreaks

\begin{document}

\date{}
\maketitle
\setstretch{1.5}

\begin{abstract}
In this paper, we present a new approach to distributed moving horizon estimation for constrained nonlinear processes. The method involves approximating the arrival costs of local estimators through a recursive framework. First, distributed full-information estimation for linear unconstrained systems is presented, which serves as the foundation for deriving the analytical expression of the arrival costs for the local estimators. Subsequently, we develop a recursive arrival cost design for linear distributed moving horizon estimation.
Sufficient conditions are derived to ensure the stability of the estimation error for constrained linear systems.
Next, we extend the arrival cost design derived for linear systems to account for nonlinear systems, and a partition-based constrained distributed moving horizon estimation algorithm for nonlinear systems is formulated. A benchmark chemical process is used to illustrate the effectiveness and superiority of the proposed method.
\end{abstract}

\noindent{\bf Keywords:} Distributed state estimation, moving horizon estimation, arrival cost approximation, nonlinear processes

\section{Introduction}
The partition-based distributed framework has  emerged as a promising structure for developing scalable and flexible decision-making solutions for large-scale complex industrial processes, since it can provide higher fault tolerance, reduced computational complexity, and increased flexibility for system maintenance
\cite{christofides2013distributed, battistelli2016stability,yin2018forming}.
Within a partition-based distributed decision-making framework, a large-scale process is partitioned into smaller subsystems that are interconnected with each other.
Multiple decision-making units are deployed for the subsystems and coordinate their decisions through real-time communication \cite{daoutidis2019decomposition,tang2018optimal,chen2020machine}.
To enable distributed decision-making systems to take informed control actions for appropriate process operation, it is crucial to have distributed state estimation capabilities that can provide real-time full-state estimates for the underlying systems
\cite{battistelli2016stability,farina2009moving,schneider2015convergence}.
In this paper, we focus on partition-based distributed state estimation for general nonlinear systems.

As an effective distributed state estimation approach, distributed moving horizon estimation (DMHE) offers the capability to handle process nonlinearity and address constraints imposed on both state variables and process disturbances \cite{farina2009moving,schneider2015convergence,schneider2017solution,schneider2015iterative,yin2022event, li2023iterative,yin2017distributed,farina2011moving,zhang2013distributed,zeng2015distributed}. In \cite{farina2009moving}, non-iterative partition-based DMHE algorithms
were proposed for linear systems considering constraints on state variables and process disturbances. These algorithms design local estimators based on partitioned subsystem models, with each local estimator handling state estimation for the corresponding subsystem with non-overlapping states.
In \cite{schneider2015convergence,schneider2017solution,schneider2015iterative,yin2022event,li2023iterative}, partition-based DMHE approaches that require iterative executions within each sampling period were proposed for linear systems; these designs ensure the convergence of the state estimates generated by DMHE to their the corresponding centralized moving horizon estimation (MHE) counterparts.
Based on these approaches, the objective function of centralized MHE is partitioned into several individual objective functions. An additional term is then incorporated with each partitioned objective function to construct the local objective function for the proposed DMHE algorithms. Particularly, in  \cite{schneider2015convergence,schneider2017solution,schneider2015iterative}, a sensitivity term is integrated with the partitioned objective function to account for the impact of each local decision variable on the objective functions of interconnected subsystems, while in \cite{yin2022event,li2023iterative}, penalties on measurement noise from interconnected subsystems are incorporated to form the local objective function of each estimator. In \cite{yin2017distributed,farina2011moving}, partition-based DMHE approaches for nonlinear systems were proposed.
In \cite{zhang2013distributed,zeng2015distributed}, DMHE for constrained nonlinear systems was addressed in a way that an auxiliary observer is integrated with the corresponding MHE to form an enhanced MHE-based constrained estimator for each subsystem of the entire nonlinear process.

In MHE design, previous information not included in the current estimation window can be summarized by a function referred to as arrival cost. {\color{black}An accurate approximation of the arrival cost can enhance estimation performance \cite{qu2009computation,lopez2011constrained}. Additionally, a well-approximated arrival cost allows for a reduction in the length of the estimation window without compromising the accuracy of state estimates \cite{lopez2011constrained}. This reduction in the estimation window length can enhance the computational efficiency by decreasing the complexity of the optimization problem.}
In centralized MHE designs, various methods have been proposed to approximate the arrival cost. For linear systems, the Kalman filter has been widely used in the approximation of arrival cost \cite{rao2001constrained,gharbi2019proximity}. For nonlinear systems, solutions for approximating the arrival cost of centralized MHE include extended Kalman filter \cite{rao2002constrained,gharbi2020proximity}, unscented Kalman filter \cite{qu2009computation}, and particle filter \cite{lopez2011constrained}.
In a distributed context, accurately approximating the arrival costs for the local estimators of DMHE becomes a more complicated problem. Different approximation methods for arrival cost approximation have been adopted for linear DMHE.
In
\cite{schneider2015convergence,schneider2017solution,schneider2015iterative,yin2022event,li2023iterative}, the arrival cost was formulated as a weighted squared error between the state estimate and the \emph{a priori} prediction, weighted by a constant matrix, which is fine-tuned to satisfy stability conditions. Additionally, in \cite{farina2009moving}, a Kalman filter design for an auxiliary system was leveraged to approximate the weighting matrix for the arrival cost at each sampling instant. In \cite{li2023data}, a partition-based DMHE method was proposed for the state estimation of data-driven subsystem models. In this work, the update of the arrival cost for DMHE design was facilitated by using a partition-based distributed Kalman filter approach proposed in \cite{li2023partition}.
Meanwhile, results on approximating the arrival costs for nonlinear DMHE algorithms have been limited.
In \cite{yin2017distributed} where a two-time-scale nonlinear DMHE was proposed, the arrival costs for the local estimators were not considered. In \cite{farina2011moving}, the weighting matrix for the arrival cost design of each estimator was updated at each sampling instant. However, this paper only presents the conditions for the weighting matrix to satisfy and does not explicitly provide the update formula for the weighting matrix.
In \cite{zhang2013distributed,zeng2015distributed}, decentralized extended Kalman filters were utilized to approximate the arrival costs for local estimators of observer-enhanced DMHE. However, in each of the two designs, the interactive dynamics were not taken into account.

In this paper, we address the problem of approximating the arrival costs for the local estimators of a partition-based DMHE design and formulate a partition-based distributed estimation scheme for general nonlinear processes with state constraints.
The objective of this work is achieved in four steps: 1) we derive an analytical expression of the arrival cost for each local estimator of the DMHE in the unconstrained linear context from the design of the distributed full-information estimation formulation in \cite{li2023partition};
2) we conduct the stability analysis for the proposed DMHE algorithm for linear systems with state constraints;
3) the analytical expression of arrival cost obtained for linear unconstrained systems is extended to account for nonlinear systems; 4) we formulate a partition-based constrained DMHE algorithm for general nonlinear systems, where each local estimator incorporates output measurements of the interacting subsystems and approximates the local arrival cost using the derived recursive solution. The effectiveness of the proposed method is demonstrated using a simulated chemical process.
Some initial findings from this study were presented in a conference paper \cite{NMPC}.
Compared with \cite{NMPC}, this paper presents the stability analysis for the proposed DMHE algorithm for linear systems with state constraints. Additionally, we include additional comparisons to demonstrate the efficacy and superiority of the proposed DMHE approach.

\section{Preliminaries}
\subsection{Notation}
$\text{diag}\left\{A_1,\ldots,A_n\right\}$ represents the block diagonal matrix with blocks
$A_{i}$, $i=1,\ldots, n$. $[A_{ij}]$ represents a block matrix where the $A_{ij}$ is the submatrix in the $i$th row and the $j$th column. $I_{n}$ is an $n\times n$ identity matrix.
$\|z\|_{A}^{2}=z^{\mathrm{T}}Az$ is the square of the weighted Euclidean norm of vector $z$. $\text{col}\{z_{1},\ldots,z_{n}\}$ denotes a column vector containing a sequence $z_{1}, \ldots, z_{n}$.
$\{z\}_{a}^{b}=\text{col}\{z_{a}, z_{a+1}, \ldots, z_{b}\}$.

\subsection{System description}\label{sec:2.2}
Let us consider nonlinear systems consisting of $n$ interconnected subsystems.
The dynamics of the $i$th subsystem, $i\in\mathbb{I}=\{1,2,\ldots,n\}$, are expressed as follows:
\begin{subequations}\label{model_nonlinear}
  \begin{align}
    x^{i}_{k+1} & = f_{i}\big(x^{i}_{k}, X^{i}_{k}\big)+w^{i}_{k} \label{model_nonlinear1} \\
    y^{i}_{k} & =h_{i}(x^{i}_{k})+v^{i}_{k}\label{model_nonlinear2}
  \end{align}
\end{subequations}
where $k$ denotes discrete-sampling instant; $x^{i}_{k}\in\mathbb{R}^{n_{x^{i}}}$ and $y^{i}_{k}\in\mathbb{R}^{n_{y^i}}$ represent the state vector and output measurements of the $i$th subsystem, respectively;
$X^{i}_{k}\in \mathbb{R}^{n_{X^{i}}}$ is a vector of the states of all the subsystems that have direct influence on the dynamics of subsystem $i$;  $w^{i}_{k}\in\mathbb{R}^{n_{x^i}}$ and $v^{i}_{k}\in\mathbb{R}^{n_{y^{i}}}$ represent the unknown process disturbances and measurement noise associated with subsystem $i$, respectively;
$f_{i}$ is a vector-value nonlinear function characterizing the dynamics of subsystem $i$; $h_i$ is the output measurement function for subsystem $i$, $i\in\mathbb{I}$.

By considering the variables for the entire system as the aggregation of the variables for each subsystem $i$, $i\in\mathbb{I}$: $z_{k}={\rm{col}}\{z^1_k,\ldots, z^n_k\}\in\mathbb{R}^{n_z}$ for vector $z\in\{x,y,w,v\}$,
the dynamics of the entire system can be formulated in the following compact form:
\begin{subequations}\label{model_nonlinear_centralized}
\begin{align}
    x_{k+1} & = f(x_{k})+w_{k} \label{model_nonlinear_centralized1} \\
    y_{k}& =h(x_{k})+v_{k}\label{model_nonlinear_centralized2}
\end{align}
\end{subequations}
where $f(x_{k})$ and $h(x_{k})$ can be obtained from $f_{i}(x_{k},X^{i}_{k})$ and $h_{i}(x^{i}_{k})$, respectively.
\subsection{Problem formulation}
In this work, we aim to propose a partition-based DMHE algorithm for nonlinear systems, which integrates a recursive update mechanism for the arrival cost of each local estimator of the distributed scheme. To achieve this goal,
we first present a distributed full-information estimation (FIE) design by partitioning the objective function of a centralized FIE problem. From this design, we derive the analytical expression of the arrival cost for the proposed DMHE method in a linear context. Subsequently, we formulate the proposed DMHE method based on the obtained arrival costs, and sufficient conditions are provided to guarantee the stability of the proposed DMHE approach under state constraints. Finally, the obtained arrival costs for linear systems are extended to account for the nonlinear systems, and a partition-based DMHE algorithm for constrained nonlinear context is formulated.

Based on the above consideration, first, we start by examining a class of linear systems consisting of
$n$ subsystems. The dynamics of $i$th linear subsystem, $i\in\mathbb{I}$, are described as follows:
\begin{subequations}\label{model}
  \begin{align}
   x^{i}_{k+1}&=A_{ii}x^{i}_k+\sum_{l\in\mathbb{I}\setminus\{i\}}A_{il}x^{l}_k+w^{i}_k\label{modeli1}\\
   y^{i}_k&=C_{ii}x^{i}_k+v^{i}_k\label{modeli2}
\end{align}
\end{subequations}
\normalsize
where $A_{ii}$, $A_{il}$, and $C_{ii}$, $i\in\mathbb{I}$, $l\in\mathbb{I}\setminus\{i\}$, are subsystem matrices of compatible dimensions.

Based on the linear subsystem models in \eqref{model}, we will design MHE-based local estimators. These estimators will then be integrated to formulate a linear DMHE design, which will be extended to account for state estimation of nonlinear systems in \eqref{model_nonlinear}.
A compact form of the linear system comprising all the subsystems in the form of \eqref{model} is described as follows:
\begin{subequations}\label{cmodel}
  \begin{align}
   x_{k+1} & =Ax_{k}+w_{k} \label{cmodel1}\\
   y_k & =Cx_k+v_k \label{cmodel2}
\end{align}
\end{subequations}
where {\color{black} $A=[A_{il}]$ represents a block matrix where $A_{il}$ is the submatrix in the $i$th row and the $l$th column;} $C=\mathrm{diag}\{C_{11}, \ldots, C_{nn}\}$.
\subsection{Centralized full-information estimation}
{\color{black}Full-information estimation (FIE) is an optimization-based state estimation approach, which can be viewed as a least-squares estimation method that minimizes the cumulative sum of squared errors from the initial time instant to the current time instant.}
Specifically, at each sampling instant $k$, based on the linear model in \eqref{cmodel}, a centralized FIE design can be formulated as follows \cite{findeisen1997moving}:
\begin{subequations}\label{cfie}
    \begin{align}
&\min_{\{\hat{x}_{j}\}_{j=0}^{j=k}} \bar{\Phi}_{k} =\sum_{j=0}^{k-1}\|\hat{w}_{j}\|^2_{Q^{-1}}+\sum_{j=0}^{k}\|\hat{v}_{j}\|^2_{R^{-1}}+\|\hat{x}_{0}-\bar{x}_{0}\|_{P_{0}^{-1}}^{2}\nonumber\\
&\quad\text{s.t.}~~~~~  \hat{x}_{j+1}=A\hat{x}_{j}+\hat{w}_{j}\\
&\quad\quad\quad~~~~~~\, y_{j} = C\hat{x}_{j}+\hat{v}_{j}
\end{align}
\end{subequations}
where $\hat{x}_j$ is an estimate of state $x_{j}$; $\bar{x}_{0}$ is an \emph{a priori} estimate of initial state $x_{0}$; $\hat{w}_j$ and $\hat{v}_j$ are estimates of process disturbances $w_j$ and measurement noise $v_j$, respectively; $P_{0}$,
$Q$, and $R$ are positive-definite weighting matrices.
{\color{black} It is noted that as the number of sampling instants increases over time, the associated optimization problem in \eqref{cfie} becomes increasingly complex and intractable.  FIE is important for the design and analysis of the MHE approaches \cite{findeisen1997moving,knufer2023nonlinear}. In the next section, we introduce a distributed FIE design, which is used to guide the development of the DMHE approach.}

\section{Distributed full-information estimation}
In this section, a distributed FIE problem is formulated based on the linear subsystem models in \eqref{model}. First, we partition the global objective function of the centralized FIE design in \eqref{cfie}, following \cite{li2023partition}. Subsequently, each partitioned objective function is integrated with the sensor measurements from interconnected subsystems to construct the local objective function of each local estimator.
\subsection{Construction of local objective functions}
Following \cite{li2023partition}, the global objective function of the centralized FIE algorithm $\bar{\Phi}_k$ is decomposed into $\bar{\Phi}^{i}_k$, $i\in\mathbb{I}$, such that $\bar{\Phi}_k=\sum_{i\in\mathbb{I}}\bar{\Phi}^{i}_{k}$. The objective function of each local estimator is described as follows:
\begin{align}\label{eq:barPhi}
  \bar{\Phi}^{i}_k &=\sum_{j=0}^{k-1}\|\hat{w}^{i}_j\|^2_{Q_{i}^{-1}}+\sum_{j=0}^{k}\|\hat{v}^{i}_j\|^2_{R_{i}^{-1}}+\|\hat{x}^{i}_0-\bar{x}^i_0\|_{P_{i,0}^{-1}}^{2}
\end{align}
\normalsize
where $\hat{x}^{i}_0$ is an estimate of the $i$th subsystem state $x^{i}_0$; $\bar{x}^{i}_0$ is an initial guess of $x^{i}_0$; $\hat{w}^{i}_j$ and $\hat{v}^{i}_j$ are the estimates of disturbances $w^{i}_j$ and measurement noise $v^{i}_j$ of the $i$th subsystem, respectively. It is noted that the sensor measurements from the interconnected subsystems can provide valuable information for estimating the local subsystem states. Inspired by the objective function designs for the distributed state estimation algorithms proposed in \cite{schneider2015convergence,schneider2015iterative,li2023partition}, each local objective function of the local estimator of distributed FIE is presented as follows:
 \begin{align}\label{eq:4}
 \Phi_{k}^{i} &= \bar{\Phi}_{k}^{i}+\sum_{l\in\mathbb{I}\setminus\{i\}}\sum_{j=0}^{k}\|\hat{v}_{j}^{l}\|^{2}_{R_{l}^{-1}}\nonumber\\
 & =\sum_{j=0}^{k-1}\|\hat{w}^{i}_{j}\|^2_{Q_{i}^{-1}}+\sum_{j=0}^{k}\|\hat{v}^{[i]}_{j}\|^2_{R^{-1}}+\|\hat{x}^{i}_{0}-\bar{x}^{i}_{0}\|_{P_{i,0}^{-1}}^{2}
\end{align}
where $\hat{v}^{[i]}_{j}$ is an estimate of measurement noise $v_{j}$ of the entire system within estimator $i$;
$P_{i,0}$, $Q_{i}$, and $R_{i}$ are the $i$th diagonal block of weighting matrices $P_{0}$, $Q$, and $R$, respectively.
{\color{black}The key difference between $\bar{\Phi}_k^i$ in \eqref{eq:barPhi} and $\Phi_k^i$ in \eqref{eq:4} is that $\Phi_k^i$ includes additional information on the sensor measurements of interconnected subsystems, which can provide valuable insights for the reconstruction of the local subsystem states.}


\subsection{Distributed full-information estimation formulation}
At sampling instant $k$, the local estimator for the $i$th subsystem of distributed FIE can be formulated by
leveraging the local objective function proposed in \eqref{eq:4}:
\normalsize
\begin{subequations}\label{fie}
\begin{align}
&\min_{\{\hat{x}^{i}_j\}_{j=0}^{j=k}} \Phi^{i}_{k}=\sum_{j=0}^{k-1}\|\hat{w}^{i}_{j}\|^2_{Q_{i}^{-1}}+\sum_{j=0}^{k}\|\hat{v}^{[i]}_{j}\|^2_{R^{-1}}+\|\hat{x}^{i}_{0}-\bar{x}^{i}_{0}\|_{P_{i,0}^{-1}}^{2} \nonumber\\
&~\,\mathrm{s.t.}~\hat{x}_{j+1}^{i}=A_{ii}\hat{x}_{j}^{i}+\sum_{l\in\mathbb{I}\setminus\{i\}}A_{il}\tilde{x}_{j}^{l}+\hat{w}_{j}^{i}\label{eq:upd0_2}\\
&~\,\quad~~\quad\, y_{0}= C_{[:,i]}\hat{x}_{0}^{i}+\sum_{l\in\mathbb{I}\setminus\{i\}}C_{[:,l]}\tilde{x}_{0}^{l}+\hat{v}_{0}^{[i]}\label{eq:upd0_4}\\
&~\,\quad~~~ y_{j+1}= C(A_{[:,i]}\hat{x}_{j}^{i}+\sum_{l\in\mathbb{I}\setminus\{i\}}A_{[:,l]}\tilde{x}_{j}^{l})+v_{j+1}^{[i]},~j=0,\ldots,k-1\label{eq:upd0_3}
\end{align}
\end{subequations}
\normalsize
where $A_{[:,i]}$ and $C_{[:,i]}$ comprise the columns of $A$ and $C$ with respect to the state of $x^{i}$, respectively;
{\color{black} $\tilde{x}_{j}^{l}$ is a (conservative) state estimate of subsystem $l$ for sampling instant $j$ made available to the estimator of subsystem $i$, which is determined as:
\begin{equation}\label{eq:tilde}\normalsize
\tilde{x}_{j}^{l} =\left\{ \begin{array}{l}
\hat{x}_{j|k-1}^{l},~~~~~{\text{for}}~ j=1,\ldots,k-1~~\\[0.3em]
\bar{x}_{0}^{l},~~~~~~~~~~{\text{for}}~ j=0~~
\end{array} \right.
\end{equation}
In \eqref{eq:tilde}, $\hat{x}_{j|k-1}^{l}$ is the state estimate of $x^{l}_{j}$ calculated by the $l$th estimator at sampling instant $k-1$, and $\bar{x}_{0}^{l}$ denotes an initial guess of $x_{0}^{l}$, $l\in\mathbb{I}\setminus\{i\}$. {\color{black} We utilize the state estimates $\hat{x}_{j|k-1}^{l}$ obtained at the previous time instant $k-1$ because they are calculated based on the most recent available sensor measurements.} In the remainder of this paper, we simplify the subscript by omitting ``$|k$", and we
denote $\hat{x}_{j|k}^{i}$ by $\hat{x}_{j}^{i}$ for brevity.
}

\section{Distributed moving horizon estimation for linear systems}
In this section, a linear DMHE design with a recursive update of arrival costs for the local estimators is presented. The stability of the proposed linear DMHE design is analyzed.
First, we obtain an analytical recursive expression of the arrival cost by leveraging the distributed FIE in \eqref{fie}. Following this, we formulate a DMHE design where the arrival cost of each local estimator is updated using the derived recursive method for linear systems with state constraints. Additionally, we prove the stability of the proposed DMHE approach.

Inspired by the DMHE designs proposed by \cite{schneider2015convergence} and \cite{li2023iterative}, the $i$th estimator of the proposed DMHE
at sampling instant $k$ solves the optimization below:
\begin{subequations}\label{eq:upd0}
\begin{equation}
\min_{\{\hat{x}_{j}^{i}\}_{j=k-N}^{k}}\Phi_{k}^{i,\mathrm{MHE}}~\text{subject to \eqref{eq:upd0_2}, \eqref{eq:upd0_4}, and \eqref{eq:upd0_3}}
\end{equation}
with
\normalsize
\begin{align}
\Phi_{k}^{i,\mathrm{MHE}}&=V_{k-N}^{i,o}+\sum_{j=k-N}^{k-1}\|\hat{w}^{i}_{j}\|^2_{Q_{i}^{-1}}+\sum_{j=k-N}^{k}\|\hat{v}_{j}^{[i]}\|^2_{R^{-1}}
\end{align}
\end{subequations}
\normalsize
where $N\geq1$ is the length of the estimation window; $V_{k-N}^{i,o}$ is the arrival cost that summarizes the historical information excluded from the estimation window. The detailed arrival cost design will be discussed in Section \ref{sec:4.1}. Before proceeding, we introduce two lemmas, which will be used to derive the expression of the arrival cost for each estimator of the proposed DMHE design.

\begin{lemma}(\cite{rawlings2012postface})\label{lem:j1+j2}
  Consider the following two quadratic functions:
  \begin{equation*}
    J_{1}(x) = \|x-a\|_{A^{-1}}^{2},\quad J_{2}(x) = \|Cx-b\|_{B^{-1}}^{2},
  \end{equation*}
  where $A$ and $B$ are positive definite matrices; a and b are two vectors of compatible dimensions. Then, it holds that
  \begin{equation*}
    J_{1}(x)+J_{2}(x)  =\|x-\sigma\|_{H^{-1}}^{2}+\pi,
  \end{equation*}
  where
\begin{align*}
  H &=A-AC^{\mathrm{T}}(CAC^{\mathrm{T}}+B)^{-1}CA,  \\
   \sigma&=a+AC^{\mathrm{T}}(CAC^{\mathrm{T}}+B)^{-1}(b-Ca), \\
   \pi&=J_{1}(\sigma)+J_{2}(\sigma).
\end{align*}
\end{lemma}

  \begin{lemma}({\color{black}Woodbury matrix identity} \cite{henderson1981deriving})\label{inversion}
If matrices $A$ and $D$ have full rank, then the following equations hold:
\begin{align*}
(A+BDC)^{-1} &= A^{-1}-A^{-1}B(CA^{-1}B+D^{-1})^{-1}CA^{-1}\\
(A+BDC)^{-1}BD&=A^{-1}B(D^{-1}+CA^{-1}B)^{-1}
\end{align*}
\end{lemma}\vspace{-1.5mm}

\normalsize
\subsection{Arrival costs for linear systems}\label{sec:4.1}
By comparing the design of the distributed FIE in \eqref{fie} with the proposed DMHE in \eqref{eq:upd0}, the arrival cost $V_{k-N}^{i,o}$ can be constructed by deriving the analytical solution to the following optimization problem:
\begin{equation*}
V_{k-N}^{i,o}=\min_{\{\hat{x}_{j}^{i}\}_{j=0}^{k-N-1}}V_{k-N}^{i}~\text{subject to \eqref{eq:upd0_2}, \eqref{eq:upd0_4}, and \eqref{eq:upd0_3}}
\end{equation*}
where
\begin{align*}\label{eq:upd00_2}
V_{k-N}^{i}&=\|\hat{x}_{0}^{i}-\bar{x}_{0}^{i}\|^{2}_{P_{i,0}^{-1}}+\sum_{j=0}^{k-N-1}\|\hat{w}^{i}_{j}\|^2_{Q_{i}^{-1}}+\sum_{j=0}^{k-N-1}\|\hat{v}_{j}^{[i]}\|^2_{R^{-1}}
\end{align*}

First, we consider the arrival cost $V_{1}^{i,o}$, which can be formulated by deriving the analytical solution to the optimization problem below:
\begin{subequations}\label{eq:upd1}
\begin{equation}\label{eq:upd1_1}
V_{1}^{i,o}=\min_{\hat{x}_{0}^{i}}V_{1}^{i}~\text{subject to \eqref{eq:upd0_2} and \eqref{eq:upd0_4} for}~j=0
\end{equation}
\normalsize
where
\begin{align}\label{eq:upd1_2}
V_{1}^{i}&=\|\hat{x}_{0}^{i}-\bar{x}_{0}^{i}\|^{2}_{P_{i,0}^{-1}}+\|\hat{v}_{0}^{[i]}\|^{2}_{R^{-1}}+\|\hat{w}_{0}^{i}\|^{2}_{Q_{i}^{-1}}
\end{align}
\end{subequations}
\normalsize
Considering \eqref{eq:upd0_4} and Lemma \ref{lem:j1+j2}, $V_{1}^{i}$ can be rewritten as
\normalsize
\begin{align}\label{eq:upd_2}
V^{i}_{1}& = \|\hat{x}_{0}^{i}-\bar{x}_{0}^{i}\|^{2}_{P_{i,0}^{-1}}+\|y_{0}-C_{[:,i]}\hat{x}_{0}^{i}-\sum_{l\in\mathbb{I}\setminus\{i\}}C_{[:,l]}\tilde{x}_{0}^{l}\|^{2}_{R^{-1}}+\|\hat{w}_{0}^{i}\|^{2}_{Q_{i}^{-1}}\nonumber\\
  & = \|\hat{x}_{0}^{i}-\breve{x}_{0}^{i}\|^{2}_{\breve{P}_{i,0}^{-1}}+\pi^{i}_{0}+\|\hat{w}_{0}^{i}\|^{2}_{Q_{i}^{-1}}
\end{align}
\normalsize
where
\normalsize
  \begin{align*}
  \breve{P}_{i,0} & = P_{i,0}- P_{i,0}C_{[:,i]}^{\mathrm{T}}(C_{[:,i]}P_{i,0}C_{[:,i]}^{\mathrm{T}}+R)^{-1}C_{[:,i]}P_{i,0}\\
   \breve{x}_{0}^{i}  & =\bar{x}_{0}^{i}+ P_{i,0}C_{[:,i]}^{\mathrm{T}}(C_{[:,i]}P_{i,0}C_{[:,i]}^{\mathrm{T}}+R)^{-1} (y_{0}- C_{[:,i]}\bar{x}_{0}^{i}-\sum_{l\in\mathbb{I}\setminus\{i\}}C_{[:,l]}\tilde{x}_{0}^{l})\\
    \pi^{i}_{0} & = \|\breve{x}_{0}^{i}  -\bar{x}_{0}^{i}\|^{2}_{P_{i,0}^{-1}}+\|y_{0}- C_{[:,i]}\breve{x}_{0}^{i}-\sum_{l\in\mathbb{I}\setminus\{i\}}C_{[:,l]}\tilde{x}_{0}^{l}\|^{2}_{R^{-1}}
  \end{align*}
\normalsize
As $\pi^{i}_{0}$ is a constant, we do not need to take it into account when solving the optimization problem \eqref{eq:upd1}.
According to \eqref{eq:upd0_2}, Lemma \ref{lem:j1+j2}, and Lemma \ref{inversion}, it is further obtained that
\begin{align}\label{eq:upd_333}
V^{i}_{1}& = \|\hat{x}_{0}^{i}-\breve{x}_{0}^{i}\|^{2}_{\breve{P}_{i,0}^{-1}}+\|\hat{x}_{1}^{i}-A_{ii}\hat{x}_{0}^{i}-\sum_{l\in\mathbb{I}\setminus\{i\}}A_{il}\tilde{x}_{0}^{l}\|^{2}_{Q_{i}^{-1}}\nonumber\\
&=\|\hat{x}_{0}^{i}-\breve{\sigma}^{i}_{0}\|^{2}_{H_{0}^{i}}+\breve{\pi}^{i}_{0}
\end{align}
where
\begin{subequations}
     \begin{align}
    H_{0}^{i}&= \breve{P}_{i,0}^{-1}+A_{ii}^{\mathrm{T}}Q_{i}^{-1}A_{ii}\\
      \breve{\sigma}^{i}_{0} & = \breve{x}_{0}^{i}+\breve{P}_{i,0}A_{ii}^{\mathrm{T}}(A_{ii}\breve{P}_{i,0}A_{ii}^{\mathrm{T}}+Q_{i})^{-1}(\hat{x}_{1}^{i}-A_{ii}\breve{x}_{0}^{i}-\sum_{l\in\mathbb{I}\setminus\{i\}}A_{il}\tilde{x}_{0}^{l}) \\
      \breve{\pi}^{i}_{0} & =
      \|\breve{x}_{0}^{i}-\breve{\sigma}^{i}_{0}\|^{2}_{\breve{P}_{i,0}^{-1}}+\|\hat{x}_{1}^{i}-A_{ii}\breve{\sigma}^{i}_{0}-\sum_{l\in\mathbb{I}\setminus\{i\}}A_{il}\tilde{x}_{0}^{l}\|^{2}_{Q_{i}^{-1}}\label{eq:p1}
    \end{align}
\end{subequations}
Let $\bar{x}_{1}^{i}:=A_{ii}\breve{x}_{0}^{i}+\sum_{l\in\mathbb{I}\setminus\{i\}}A_{il}\tilde{x}_{0}^{l}$ and $L_{i,0}=\breve{P}_{i,0}A_{ii}^{\mathrm{T}}(A_{ii}\breve{P}_{i,0}A_{ii}^{\mathrm{T}}+Q_{i})^{-1}$.
$\breve{\pi}^{i}_{0}$ in \eqref{eq:p1} can be reformulated as follows:
\begin{align*}
    \breve{\pi}^{i}_{0} & =\|L_{i,0}(\hat{x}_{1}^{i}-\bar{x}_{1}^{i})\|^{2}_{\breve{P}_{i,0}^{-1}}+\|(I-A_{ii}L_{i,0})(\hat{x}_{1}^{i}-\bar{x}_{1}^{i})\|^{2}_{Q_{i}^{-1}}\nonumber\\
    &=\|\hat{x}_{1}^{i}-\bar{x}_{1}^{i}\|^{2}_{P_{i,1}^{-1}}
\end{align*}
where
\begin{align}\label{eq:p1_1}
    P_{i,1}^{-1} &= L_{i,0}^{\mathrm{T}}\breve{P}_{i,0}^{-1}L_{i,0}+(I-A_{ii}L_{i,0})^{\mathrm{T}}Q_{i}^{-1}(I-A_{ii}L_{i,0})\nonumber\\
    &=L_{i,0}^{\mathrm{T}}(\breve{P}_{i,0}^{-1}+A_{ii}^{\mathrm{T}}Q_{i}^{-1}A_{ii})L_{i,0}-Q_{i}^{-1}A_{ii}L_{i,0}- L_{i,0}^{\mathrm{T}}A_{ii}^{\mathrm{T}}Q_{i}^{-1}+Q_{i}^{-1}
\end{align}
According to Lemma \ref{inversion}, we have
\begin{equation}\label{eq:l}
    L_{i,0} = (\breve{P}_{i,0}^{-1}+A_{ii}^{\mathrm{T}}Q_{i}^{-1}A_{ii})^{-1}A_{ii}^{\mathrm{T}}Q_{i}^{-1}
\end{equation}
Then, the first term on the right-hand-side of \eqref{eq:p1_1} can be rewritten as:
    \begin{align}\label{eq:lpl}
        &\quad L_{i,0}^{\mathrm{T}}(\breve{P}_{i,0}^{-1}+A_{ii}^{\mathrm{T}}Q_{i}^{-1}A_{ii})L_{i,0}\nonumber\\
        &=Q_{i}^{-1}A_{ii}(\breve{P}_{i,0}^{-1}+A_{ii}^{\mathrm{T}}Q_{i}^{-1}A_{ii})^{-1}(\breve{P}_{i,0}^{-1}+A_{ii}^{\mathrm{T}}Q_{i}^{-1}A_{ii})(\breve{P}_{i,0}^{-1}+A_{ii}^{\mathrm{T}}Q_{i}^{-1}A_{ii})^{-1}A_{ii}^{\mathrm{T}}Q_{i}^{-1}\nonumber\\
        &=Q_{i}^{-1}A_{ii}(\breve{P}_{i,0}^{-1}+A_{ii}^{\mathrm{T}}Q_{i}^{-1}A_{ii})^{-1}A_{ii}^{\mathrm{T}}Q_{i}^{-1},
    \end{align}
Substituting \eqref{eq:l} and \eqref{eq:lpl} into \eqref{eq:p1_1} yields
\begin{align}\label{eq:p1_2}
    P_{i,1}^{-1} &= Q_{i}-Q_{i}^{-1}A_{ii}(\breve{P}_{i,0}^{-1}+A_{ii}^{\mathrm{T}}Q_{i}^{-1}A_{ii})^{-1}A_{ii}^{\mathrm{T}}Q_{i}^{-1}\nonumber\\
    & = (Q_{i}+A_{ii}\breve{P}_{i,0}A_{ii}^{\mathrm{T}})^{-1}
\end{align}
Therefore, it is further obtain $P_{i,1}=Q_{i}+A_{ii}\breve{P}_{i,0}A_{ii}^{\mathrm{T}}$.
By minimizing $V^{i}_{1}$ in \eqref{eq:upd_333} with respect to $\hat{x}_{0}^{i}$, the arrival cost $V^{i,o}_{1}$ for the $i$th estimator can be derived as follows:
\begin{equation*}\label{eq:upd6}
  V_{1}^{i,o} = \min_{\hat{x}_{0}^{i}}V_{1}^{i}=\|\hat{x}_{1}^{i}-\bar{x}_{1}^{i}\|^{2}_{P_{i,1}^{-1}}
\end{equation*}

Next, let us proceed to the arrival cost $V_{2}^{i,o}$ by addressing the following optimization problem:
\normalsize
\begin{subequations}\label{eq:upd7}
\begin{equation}\label{eq:upd7_1}
V_{2}^{i,o}=\min_{\{\hat{x}_{j}^{i}\}_{j=0}^{j=1}}V_{2}^{i}~\text{subject to \eqref{eq:upd0_2}, \eqref{eq:upd0_4}, and\eqref{eq:upd0_3} for}~j=0,1
\end{equation}
\normalsize
where
\normalsize
\begin{align}\label{eq:upd7_2}
V_{2}^{i}&=\|\hat{x}_{0}^{i}-\bar{x}_{0}^{i}\|^{2}_{P_{i,0}^{-1}}+\|\hat{v}_{0}^{[i]}\|^{2}_{R^{-1}}+\|\hat{w}_{0}^{i}\|^{2}_{Q_{i}^{-1}}+\|\hat{v}_{1}^{[i]}\|^{2}_{R^{-1}}+\|\hat{w}_{1}^{i}\|^{2}_{Q_{i}^{-1}}
\end{align}
\end{subequations}
\normalsize
By following the same procedure adopted to derive \eqref{eq:upd_2}, $V_{2}^{i}$ can be expressed as:
\begin{align*}
  V_{2}^{i} &=  \|\hat{x}_{0}^{i}-\breve{x}_{0}^{i}\|^{2}_{\breve{P}_{i,0}^{-1}}+\|\hat{v}_{1}^{[i]}\|^{2}_{R^{-1}}+\|\hat{w}_{0}^{i}\|^{2}_{Q_{i}^{-1}}+\|\hat{w}_{1}^{i}\|^{2}_{Q_{i}^{-1}},
\end{align*}
\normalsize
up to a constant term.
Based on \eqref{eq:upd0_3}, it is further obtained that
\normalsize
\begin{align*}
  V_{2}^{i}& = \|\hat{x}_{0}^{i}-\breve{x}_{0}^{i}\|^{2}_{\breve{P}_{i,0}^{-1}}+\|y_{1}-C(A_{[:,i]}\hat{x}_{0}^{i}+\sum_{l\in\mathbb{I}\setminus\{i\}}A_{[:,l]}\tilde{x}_{0}^{l})\|^{2}_{R^{-1}}
 +\|\hat{w}_{0}^{i}\|^{2}_{Q_{i}^{-1}}+\|\hat{w}_{1}^{i}\|^{2}_{Q_{i}^{-1}}\\
&=\|\hat{x}_{0}^{i}-\check{x}_{0}^{i}\|^{2}_{\check{P}_{i,0}^{-1}}+\check{\pi}_{0}^{i}+\|\hat{w}_{0}^{i}\|^{2}_{Q_{i}^{-1}}+\|\hat{w}_{1}^{i}\|^{2}_{Q_{i}^{-1}}
\end{align*}
\normalsize
where
\normalsize
\begin{align*}
   \check{P}_{i,0} & = \breve{P}_{i,0}- \breve{P}_{i,0}A_{[:,i]}^{\mathrm{T}}C^{\mathrm{T}}(CA_{[:,i]}\breve{P}_{i,0}A_{[:,i]}^{\mathrm{T}}C^{\mathrm{T}}+R)^{-1}CA_{[:,i]}\breve{P}_{i,0}\\
   \check{x}_{0}^{i}& =  \breve{x}_{0}^{i}+\breve{P}_{i,0}A_{[:,i]}^{\mathrm{T}}C^{\mathrm{T}}(CA_{[:,i]}\breve{P}_{i,0}A_{[:,i]}^{\mathrm{T}}C^{\mathrm{T}}+R)^{-1}(y_{1}-C(A_{[:,i]}\breve{x}_{0}^{i}+\sum_{l\in\mathbb{I}\setminus\{i\}}A_{[:,l]}\tilde{x}_{0}^{l}))\\
   \check{\pi}_{0}^{i}& = \|\check{x}_{0}^{i}-\breve{x}_{0}^{i}\|^{2}_{\breve{P}_{i,0}^{-1}}+\|y_{1}-C(A_{[:,i]}\check{x}_{0}^{i}+\sum_{l\in\mathbb{I}\setminus\{i\}}A_{[:,l]}\tilde{x}_{0}^{l})\|^{2}_{R^{-1}}
\end{align*}
\normalsize
$\check{\pi}_{0}^{i}$ is a constant and is neclected in deriving $V_{2}^{i}$. Then, by applying the same procedure used to obtain \eqref{eq:upd_333} and \eqref{eq:p1_2}, it is derived that:
\normalsize
\begin{align}\label{eq:upd10}
  V_{2}^{i}& =\|\hat{x}_{0}^{i}-\check{x}_{0}^{i}\|^{2}_{\check{P}_{i,0}^{-1}}+\|\hat{x}_{1}^{i}-A_{ii}\hat{x}_{0}^{i}-\sum_{l\in\mathbb{I}\setminus\{i\}}A_{il}\tilde{x}_{0}^{l}\|^{2}_{Q_{i}^{-1}}+\|\hat{w}_{1}^{i}\|^{2}_{Q_{i}^{-1}}\nonumber\\
  & = \|\hat{x}_{0}^{i}-\bar{\sigma}_{0}^{i}\|^{2}_{\bar{H}_{0}^{i}}+\bar{\pi}_{0}^{i}+\|\hat{w}_{1}^{i}\|^{2}_{Q_{i}^{-1}}
\end{align}
\normalsize
where
\normalsize
    \begin{align*}
\bar{H}_{0}^{i}&=\check{P}_{i,0}^{-1}+A_{ii}^{\mathrm{T}}Q_{i}^{-1}A_{ii}\\
      \bar{\sigma}^{i}_{0} & = \check{x}_{0}^{i}+\check{P}_{i,0}A_{ii}^{\mathrm{T}}(A_{ii}\check{P}_{i,0}A_{ii}^{\mathrm{T}}+Q_{i})^{-1}(\hat{x}_{1}^{i}-A_{ii}\check{x}_{0}^{i}-\sum_{l\in\mathbb{I}\setminus\{i\}}A_{il}\tilde{x}_{0}^{l}) \\
      \breve{P}_{i,1}  & =Q_{i}+A_{ii}\check{P}_{i,0}A_{ii}^{\mathrm{T}} \\
      \bar{\pi}^{i}_{0} & =
    \|\hat{x}_{1}^{i}-A_{ii}\check{x}_{0}^{i}-\sum_{l\in\mathbb{I}\setminus\{i\}}A_{il}\tilde{x}_{0}^{l}\|^{2}_{
    \breve{P}_{i,1}^{-1}}
\end{align*}
\normalsize
Let us further define $\breve{x}_{1}^{i} := A_{ii}\check{x}_{0}^{i}+\sum_{l\in\mathbb{I}\setminus\{i\}}A_{il}\tilde{x}_{0}^{l}$. Based on \eqref{eq:upd10}, we apply the forward dynamic programming method to the optimization problem \eqref{eq:upd7}, and it is obtained that
\begin{align} \label{eq:upd11}
   V_{2}^{i,o} &= \min_{\hat{x}_{1}^{i}}\big\{\min_{\hat{x}_{0}^{i}}\{\|\hat{x}_{0}^{i}-\bar{\sigma}_{0}^{i}\|^{2}_{\bar{H}_{0}^{i}}\}+\|\hat{x}_{1}^{i}-\breve{x}_{1}^{i}\|^{2}_{
    \breve{P}_{i,1}^{-1}}+\|\hat{w}_{1}^{i}\|^{2}_{Q_{i}^{-1}}\big\}~\text{subject to \eqref{eq:upd0_2} for}~j=1
  \end{align}
\normalsize
It is noted that the optimal value of $\hat{x}_{0}^{i}$ for optimization problem \eqref{eq:upd11} is $\bar{\sigma}_{0}^{i}$. Therefore, it is further derived
\begin{align}\label{eq:upd111}
   V_{2}^{i,o} &= \min_{\hat{x}_{1}^{i}}\big\{\|\hat{x}_{1}^{i}-\breve{x}_{1}^{i}\|^{2}_{
    \breve{P}_{i,1}^{-1}}+\|\hat{w}_{1}^{i}\|^{2}_{Q_{i}^{-1}}\big\}~\text{subject to \eqref{eq:upd0_2} for}~j=1
  \end{align}
\normalsize
In the following, according to \eqref{eq:upd0_2}, the objective function of \eqref{eq:upd111} can be rewritten as
\begin{align}\label{eq:18}
    V_{2}^{i}&= \|\hat{x}_{1}^{i}-\breve{x}_{1}^{i}\|^{2}_{\breve{P}_{i,1}^{-1}}+\|\hat{x}_{2}^{i}-A_{ii}\hat{x}_{1}^{i}-\sum_{l\in\mathbb{I}\setminus\{i\}}A_{il}\tilde{x}_{1}^{l}\|^{2}_{Q_{i}^{-1}}\nonumber\\
     & = \|\hat{x}_{1}^{i}-\breve{\sigma}_{1}^{i}\|^{2}_{H_{1}^{i}}+\breve{\pi}_{1}^{i}
\end{align}
 \normalsize
where
\normalsize
    \begin{align*}
H_{1}^{i}&=\breve{P}_{i,1}^{-1}+A_{ii}^{\mathrm{T}}Q_{i}^{-1}A_{ii}\\
      \breve{\sigma}^{i}_{1} & = \breve{x}_{1}^{i}+\breve{P}_{i,1}A_{ii}^{\mathrm{T}}(A_{ii}\breve{P}_{i,1}A_{ii}^{\mathrm{T}}+Q_{i})^{-1}(\hat{x}_{2}^{i}-A_{ii}\breve{x}_{1}^{i}-\sum_{l\in\mathbb{I}\setminus\{i\}}A_{il}\tilde{x}_{1}^{l}) \\
      P_{i,2}  & =Q_{i}+A_{ii}\breve{P}_{i,1}A_{ii}^{\mathrm{T}} \\
      \breve{\pi}^{i}_{1} & =
    \|\hat{x}_{2}^{i}-A_{ii}\breve{x}_{1}^{i}-\sum_{l\in\mathbb{I}\setminus\{i\}}A_{il}\tilde{x}_{1}^{l}\|^{2}_{
    P_{i,2}^{-1}}
\end{align*}
\normalsize
Let $\bar{x}_{2}^{i}:=A_{ii}\breve{x}_{1}^{i}+\sum_{l\in\mathbb{I}\setminus\{i\}}A_{il}\tilde{x}_{1}^{l}$. By minimizing $V_{2}^{i}$ in \eqref{eq:18} with respect to $\hat{x}_{1}^{i}$, the arrival cost $V_{2}^{i,o}$ in \eqref{eq:upd111} of the $i$th estimator can be obtained as:
\begin{equation*}
  V_{2}^{i,o} =\min_{\hat{x}_{1}^{i}}V_{2}^{i}=\|\hat{x}_{2}^{i}-\bar{x}_{2}^{i}\|^{2}_{P_{i,2}^{-1}}
\end{equation*}

\begin{figure}
  \centering
  \includegraphics[width=1\textwidth]{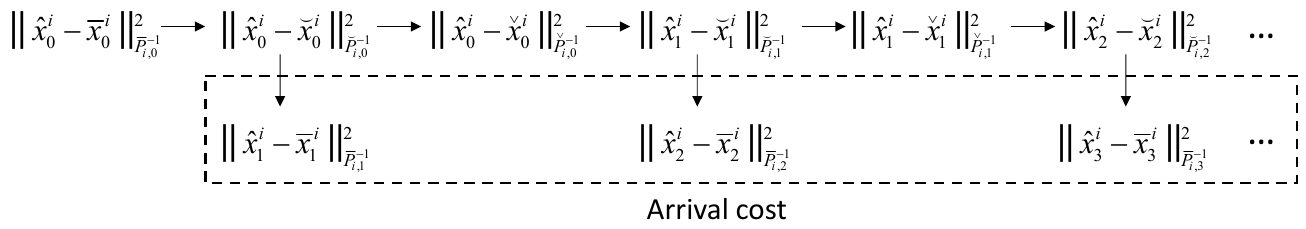}
  \caption{A flowchart of procedures for obtaining the arrival cost}\label{matrix}
\end{figure}

By iteratively applying the same procedure, the arrival cost $V_{k}^{i,o}$ for the subsequent sampling instants can be derived. Specifically, the recursive expression of the arrival cost $V_{k}^{i,o}$ for the $i$th subsystem is presented as follows:
\begin{equation}\label{eq:upd:arrival cost_1}
  V_{k}^{i,o} =\|\hat{x}_{k}^{i}-\bar{x}_{k}^{i}\|^{2}_{P_{i,k}^{-1}}
\end{equation}
where the update of arrival cost has three steps, as illustrated in Figure \ref{matrix}. The details of the three steps are as follows:
\begin{subequations}\label{eq:upd:arrival cost_2}
    \begin{itemize}
  \item From $\|\hat{x}_{k}^{i}-\breve{x}_{k}^{i}\|^{2}_{\breve{P}_{i,k}^{-1}}$ to $\|\hat{x}_{k}^{i}-\check{x}_{k}^{i}\|^{2}_{\check{P}_{i,k}^{-1}}$:
  \normalsize
    \begin{align}
   \check{P}_{i,k} & = \breve{P}_{i,k}- \breve{P}_{i,k}A_{[:,i]}^{\mathrm{T}}C^{\mathrm{T}}(CA_{[:,i]}\breve{P}_{i,k}A_{[:,i]}^{\mathrm{T}}C^{\mathrm{T}}+R)^{-1}CA_{[:,i]}\breve{P}_{i,k}\label{eq:20_4}\\
   \check{x}_{k}^{i}& = \breve{x}_{k}^{i}+ \breve{P}_{i,k}A_{[:,i]}^{\mathrm{T}}C^{\mathrm{T}}(CA_{[:,i]}\breve{P}_{i,k}A_{[:,i]}^{\mathrm{T}}C^{\mathrm{T}}+R)^{-1}(y_{k}-C(A_{[:,i]}\breve{x}_{k-1}^{i}+\sum_{l\in\mathbb{I}\setminus\{i\}}A_{[:,l]}\tilde{x}_{k-1}^{l}))\label{eq:20_1}
\end{align}
        \normalsize
  \item From $\|\hat{x}_{k}^{i}-\check{x}_{k}^{i}\|^{2}_{\check{P}_{i,k}^{-1}}$ to $\|\hat{x}_{k+1}^{i}-\breve{x}_{k+1}^{i}\|^{2}_{\breve{P}_{i,k+1}^{-1}}$:
  \normalsize
  \begin{align}
    \breve{P}_{i,k+1} & =Q_{i}+A_{ii}\check{P}_{i,k}A_{ii}^{\mathrm{T}}\label{eq:20_5}\\
    \breve{x}_{k+1}^{i} &=A_{ii}\check{x}_{k}^{i}+\sum_{l\in\mathbb{I}\setminus\{i\}}A_{il}\tilde{x}^{l}_{k}\label{eq:20_2}
  \end{align}
  \normalsize
  \item From $\|\hat{x}_{k}^{i}-\breve{x}_{k}^{i}\|^{2}_{\breve{P}_{i,k}^{-1}}$ to the arrival cost $\|\hat{x}_{k+1}^{i}-\bar{x}_{k+1}^{i}\|^{2}_{{P}_{i,k+1}^{-1}}$:
  \normalsize
  \begin{align}
   P_{i,k+1}  & =Q_{i}+A_{ii}\breve{P}_{i,k}A_{ii}^{\mathrm{T}}\label{eq:20_6}\\
   \bar{x}_{k+1}^{i}  & = A_{ii}\breve{x}_{k}^{i}+\sum_{l\in\mathbb{I}\setminus\{i\}}A_{il}\tilde{x}^{l}_{k}\label{eq:20_3}
  \end{align}
\end{itemize}
\end{subequations}
\normalsize

By leveraging the arrival cost obtained in \eqref{eq:upd:arrival cost_2}, the proposed DMHE design for the linear system in \eqref{cmodel} is completed. Specifically, at sampling instant $k$, the $i$th estimator of the proposed DMHE algorithm solves an optimization problem as follows:
\begin{subequations}\label{dmhe_revised_1}
\begin{align}
&\min_{\{\hat{x}^{i}_j\}_{j=k-N}^{k}} \sum_{j=k-N}^{k-1}\|\hat{w}^{i}_{j}\|^2_{Q_{i}^{-1}}+\sum_{j=k-N}^{k}\|\hat{v}^{[i]}_{j}\|^2_{R^{-1}}+\|\hat{x}^{i}_{k-N}-\bar{x}^{i}_{k-N}\|_{P_{i,k-N}^{-1}}^{2} \nonumber\\
&\quad\quad\mathrm{s.t.}~\hat{x}_{j+1}^{i}=A_{ii}\hat{x}_{j}^{i}+\sum_{l\in\mathbb{I}\setminus\{i\}}A_{il}\tilde{x}_{j}^{l}+\hat{w}_{j}^{i}\label{eq:dmhe_1_1}\\
&\quad\quad\quad~~ y_{k-N}= C_{[:,i]}\hat{x}_{k-N}^{i}+\sum_{l\in\mathbb{I}\setminus\{i\}}C_{[:,l]}\tilde{x}_{k-N}^{l}+\hat{v}_{k-N}^{[i]}\label{eq:dmhe_2_1}\\
&\quad\quad\quad~~ y_{j+1}= C(A_{[:,i]}\hat{x}_{j}^{i}+\sum_{l\in\mathbb{I}\setminus\{i\}}A_{[:,l]}\tilde{x}_{j}^{l})+v_{j+1}^{[i]},~j=k-N,\ldots,k-1\label{eq:dmhe_3_1}\\
&\quad\quad\quad~~ \hat{x}_{k}^{i}\in\mathbb{X}^{i}\label{eq:dmhe_4_1}
\end{align}
\end{subequations}
{\color{black}
In \eqref{dmhe_revised_1}, $\|\hat{w}^{i}_{j}\|^2_{Q_{i}^{-1}}$ penalizes the deviation of estimated subsystem states from the nominal process subsystem model over an estimation window, ensuring that the state estimates can follow the system dynamics. $\|\hat{v}^{[i]}_{j}\|^2_{R^{-1}}$  penalizes the discrepancy between the predicted measurements and actual measurements to ensure that the estimated states comply with measured outputs. The arrival cost $\|\hat{x}^{i}_{k-N}-\bar{x}^{i}_{k-N}\|_{P_{i,k-N}^{-1}}^{2}$ summarizes the previous information that is not considered within the current estimation window.}

\subsection{Stability analysis}
In this section, inspired by \cite{farina2024moving}, we perform the stability analysis for the proposed DMHE method in \eqref{dmhe_revised_1}, where the arrival costs of local estimators are approximated using a recursive approach for linear systems in \eqref{cmodel} with state constraints. Before proceeding further, we introduce several matrices and one lemma as follows:

\begin{align*}
    A_{i}^{*} &= \left[\begin{array}{c}
         A_{1i}\\
        \vdots\\
         A_{ni}
    \end{array}\right],
    \tilde{A}_{i}=\left[\begin{array}{ccccccc}
       A_{11}&\cdots&A_{1,i-1}&0&A_{1,i+1}&\cdots  &A_{1n}  \\
        \vdots & &\vdots&\vdots&\vdots& &\vdots\\
        A_{n1}&\cdots&A_{n,i-1}&0&A_{n,i+1}&\cdots  &A_{nn}  \\
    \end{array}\right],
    C_{i}^{*} = \left[\begin{array}{c}
         0\\
        \vdots\\
        C_{ii}\\
        \vdots\\
        0
    \end{array}\right],\\
    \tilde{C}_{i}&=\mathrm{diag}\{C_{11},\ldots,C_{i-1,i-1},0,C_{i+1,i+1},\ldots,C_{nn}\}
\end{align*}

\begin{lemma}\label{lem3}(\cite{farina2024moving})
If the rank of matrix $C$ is equal to the dimension of vector $a$, then it holds that
\begin{equation*}
    \|Cx-a\|^{2}_{A^{-1}} = \|x-b\|^{2}_{C^{\mathrm{T}}A^{-1}C}
\end{equation*}
where $b=(C^{\mathrm{T}}A^{-1}C)^{-1}C^{\mathrm{T}}A^{-1}a$.
\end{lemma}

To establish the stability of the proposed DMHE algorithm, we modify the objective function of \eqref{dmhe_revised_1} by incorporating a constant term $\Phi^{i,*}_{k-1}$ and rewrite the $i$th estimator of the proposed DMHE method in \eqref{dmhe_revised_1}.
Specifically,
at sampling instant $k$, the $i$th estimator of the proposed DMHE approach in \eqref{dmhe_revised_1} can be written as:
\begin{subequations}\label{dmhe_revised}
\begin{align}
&\quad\quad\Phi^{i,*}_{k}=\min_{\{\hat{x}^{i}_j\}_{j=k-N}^{k}} \Phi^{i,\mathrm{MHE}}_{k}\nonumber\\
&\mathrm{s.t.}~\,\hat{x}_{j+1}^{i}=A_{ii}\hat{x}_{j}^{i}+\sum_{l\in\mathbb{I}\setminus\{i\}}A_{il}\tilde{x}_{j}^{l}+\hat{w}_{j}^{i}\label{eq:dmhe_1}\\
&\quad~~ y_{k-N}= C^{*}_{i}\hat{x}_{k-N}^{i}+\tilde{C}_{i}\tilde{x}_{k-N}+\hat{v}_{k-N}^{[i]}\label{eq:dmhe_2}\\
&\quad~~~\, y_{j+1}= C(A^{*}_{i}\hat{x}_{j}^{i}+\tilde{A}_{i}\tilde{x}_{j})+\hat{v}_{j+1}^{[i]},~j=k-N,\ldots,k-1\label{eq:dmhe_3}\\
&\quad~~\quad~ \hat{x}_{k}^{i}\in\mathbb{X}^{i}\label{eq:dmhe_4}
\end{align}
where
\begin{equation}\label{eq:22e}
    \Phi^{i,\mathrm{MHE}}_{k}=\sum_{j=k-N}^{k-1}\|\hat{w}^{i}_{j}\|^2_{Q_{i}^{-1}}+\sum_{j=k-N}^{k}\|\hat{v}^{[i]}_{j}\|^2_{R^{-1}}+\|\hat{x}^{i}_{k-N}-\bar{x}^{i}_{k-N}\|_{P_{i,k-N}^{-1}}^{2} +\Phi^{i,*}_{k-1}
\end{equation}
\end{subequations}
where $\Phi^{i,*}_{k-1}$ is the optimal value of the $i$th estimator of the proposed DMHE approach calculated at sampling instant $k-1$.
{\color{black}The inclusion of the constant $\Phi^{i,*}_{k-1}$ ensures that the sequence $\{\Phi_{k}^{i,*}\}$ is non-decreasing.
By leveraging the monotonicity and boundedness of $\{\Phi_{k}^{i,*}\}$, the sequence $\{\Phi_{k}^{i,*}\}$ is convergent, as demonstrated in the proof of Proposition \ref{lem:6}. This convergence forms the basis for proving the stability of the proposed linear DMHE algorithm in Theorem \ref{theorem}.}
It is noted that $\Phi^{i,*}_{k-1}$ is a known value at sampling instant $k$ and can be disregarded when solving the optimization problem \eqref{dmhe_revised}. Therefore, the reformulated DMHE method in \eqref{dmhe_revised} is equivalent to the proposed DMHE design in \eqref{dmhe_revised_1}.

By creating augmented vectors $\hat{x}_{j}=\mathrm{col}\{\hat{x}_{j}^{1},\ldots,\hat{x}_{j}^{n}\}$, $\tilde{x}_{j}=\mathrm{col}\{\tilde{x}_{j}^{1},\ldots,\tilde{x}_{j}^{n}\}$, $\hat{w}_{j}=\mathrm{col}\{\hat{w}_{j}^{1},\ldots,\hat{w}_{j}^{n}\}$, $Y_{j}=\mathrm{col}\{\underbrace{y_{j},\ldots, y_{j}}_{n}\}$, and $\hat{V}_{j}=\mathrm{col}\{\hat{v}_{j}^{[1]},\ldots$, $\hat{v}_{j}^{[n]}\}$. The subsystem model in \eqref{eq:dmhe_1}-\eqref{eq:dmhe_3} can be concatenated to form a compact collective estimation model:
\begin{subequations}\label{collective_estimation}
\begin{align}
\hat{x}_{j+1}&=A_{d}\hat{x}_{j}+A_{r}\tilde{x}_{j}+\hat{w}_{j}\label{eq:23a}\\
Y_{k-N}& = C^{*}\hat{x}_{k-N}+\tilde{C}\tilde{x}_{k-N}+\hat{V}_{k-N}\\
Y_{j+1}&= \mathbf{C}(A^{*}\hat{x}_{j}+\tilde{A}\tilde{x}_{j})+\hat{V}_{j+1}, ~j=k-N,\ldots,k-1\label{eq:23c}
\end{align}
\end{subequations}
where $A_{d}=\mathrm{diag}\{A_{11}, \ldots,A_{nn}\}$; $A_{r}=A-A_{d}$; $C^{*}=\mathrm{diag}\{C^{*}_{1}, \ldots,C^{*}_{n}\}$; $\tilde{C}=\mathrm{col}\{\tilde{C}_{1},\ldots,\tilde{C}_{n}\}$; $\mathbf{C}=\mathrm{diag}\{\underbrace{C,\ldots,C}_{n}\}$, $A^{*}=\mathrm{diag}\{A^{*}_{1}, \ldots,A^{*}_{n}\}$; $\tilde{A}=\mathrm{col}\{\tilde{A}_{1},\ldots,\tilde{A}_{n}\}$.
We further define a collective form of the objective function by summarizing the local objective functions $\Phi^{i,\mathrm{MHE}}_{k}$ for all $i\in\mathbb{I}$:
\begin{align}\label{eq:collective_24}
    \Phi_{k}^{\mathrm{MHE}}&=\sum_{i\in\mathbb{I}}\Phi^{i,\mathrm{MHE}}_{k}\nonumber\\
    &=\|\hat{x}_{k-N}-\bar{x}_{k-N}\|^{2}_{P_{k-N}^{-1}}+\sum_{j=k-N}^{k-1}\|\hat{w}_{j}\|^{2}_{Q^{-1}}+\sum_{j=k-N}^{k}\|\hat{V}_{j}\|^{2}_{\mathbf{R}^{-1}}+\Phi^{*}_{k-1}
\end{align}
where $P_{k-N}=\mathrm{diag}\{P_{1,k-N},\ldots,P_{n,k-N}\}$; $Q=\mathrm{diag}\{Q_{1}, \ldots,Q_{n}\}$; $\mathbf{R}=\mathrm{diag}\{\underbrace{R, \ldots,R}_{n}\}$; $\Phi^{*}_{k-1}=\sum_{i\in\mathbb{I}}\Phi^{i,*}_{k-1}$.
It is worth mentioning that the solution of the proposed DMHE design in \eqref{dmhe_revised} is equivalent to solving the following optimization problem:
\begin{align}\label{eq:collective_25}
    \Phi^{*}_{k}&=\min_{\{\hat{x}_j\}_{j=k-N}^{k}}\Phi^{\mathrm{MHE}}_{k}~\text{subject to \eqref{collective_estimation} and } \hat{x}\in\mathbb{X}
\end{align}

Similar to the stability analysis conducted in \cite{farina2024moving,farina2011moving}, for a sequence $z_{j}^{i}\in\mathbb{R}^{n_{x^{i}}}$, $j=k-N+1,\ldots,k$, we define the transit cost $\Phi^{i,\mathrm{MHE}}_{[k-N+1,k]/k}(\{z_{j}^{i}\}_{j=k-N+1}^{k})$ of the proposed DMHE approach in \eqref{dmhe_revised} for the subsystem $i$:
\begin{align}
    \Phi^{i,\mathrm{MHE}}_{[k-N+1,k]/k}(\{z_{j}^{i}\}_{j=k-N+1}^{k})&=\min_{\{\hat{x}^{i}_j\}_{j=k-N}^{k}}\Phi^{i,\mathrm{MHE}}_{k}~\text{subject to \eqref{dmhe_revised} and~}\hat{x}_{j}^{i}=z_{j}^{i} ~\text{for}~ j=k-N+1,\ldots,k
\end{align}
Let $z_{j}=\mathrm{col}\{z_{j}^{1},\ldots,z_{j}^{n}\}\in\mathbb{R}^{n_{x}}$, the collective transit cost of  \eqref{eq:collective_25} takes the following form:
\begin{align*}
    \Phi^{\mathrm{MHE}}_{[k-N+1,k]/k}(\{z_{j}\}_{j=k-N+1}^{k})&=\min_{\{\hat{x}_j\}_{j=k-N}^{k}}\Phi^{\mathrm{MHE}}_{k}~\text{subject to \eqref{collective_estimation} and}~ \hat{x}_{j}=z_{j}~\text{for}~ j=k-N+1,\ldots,k
\end{align*}
such that
\begin{equation*}
    \Phi^{\mathrm{MHE}}_{[k-N+1,k]/k}(\{z_{j}\}_{j=k-N+1}^{k})=\sum_{i\in\mathbb{I}}\Phi^{i,\mathrm{MHE}}_{[k-N+1,k]/k}(\{z_{j}^{i}\}_{j=k-N+1}^{k})
\end{equation*}

Before introducing Proposition \ref{lem:4} and Lemma \ref{Lemma:14}, an unconstrained DMHE design is formulated as follows:
\begin{align}\label{dmhe_u}
\Phi^{i,u}_{k}&=\min_{\{\hat{x}^{i}_j\}_{j=k-N}^{k}} \Phi^{i,\mathrm{MHE}}_{k}~\text{subject to \eqref{eq:dmhe_1}-\eqref{eq:dmhe_3}}
\end{align}
The associated transit cost of the unconstrained DMHE approach in \eqref{dmhe_u} for the $i$th subsystem is described as follows:
\begin{align}\label{eq:transit_u}
    \Phi^{i,u}_{[k-N+1,k]/k}(\{z_{j}^{i}\}_{j=k-N+1}^{k})&=\min_{\{\hat{x}^{i}_j\}_{j=k-N}^{k}}\Phi^{i,\mathrm{MHE}}_{k}~\text{subject to \eqref{eq:dmhe_1}-\eqref{eq:dmhe_3}}\nonumber\\
    &\quad\text{and}~ \hat{x}_{j}^{i}=z_{j}^{i}~\text{for}~ j=k-N+1,\ldots,k
\end{align}

\begin{proposition}\label{lem:4}
    Let $\{\hat{x}_{j}^{i,u}\}_{j=k-N}^{k}$ be the solution of the unconstrained DMHE problem in \eqref{dmhe_u} for the subsystem $i$. Then there exist a positive-define matrix $H^{i}_{k}$, such that the transit cost $\Phi^{i,u}_{[k-N+1,k]/k}$ $(\{z_{j}^{i}\}_{j=k-N+1}^{k})$ in \eqref{eq:transit_u} is given by the following equation:
\begin{equation*}
    \Phi^{i,u}_{[k-N+1,k]/k}(\{z_{j}^{i}\}_{j=k-N+1}^{k})=\|\{z_{j}^{i}\}_{j=k-N+1}^{k}-\{\hat{x}_{j}^{i,u}\}_{j=k-N+1}^{k}\|^{2}_{H^{i}_{k}}+\Phi^{i,u}_{k}
\end{equation*}
\end{proposition}
\begin{proof}
Based on \eqref{eq:22e}, substituting the constraints \eqref{eq:dmhe_1}-\eqref{eq:dmhe_3} and the constraints $\hat{x}_{j}^{i}=z_{j}^{i}$ for $j=k-N+1, \ldots,k$, into $\Phi_{k}^{i,\mathrm{MHE}}$ yields:
    \begin{align}\label{eq:30}
        &\sum_{j=k-N+1}^{k-1}\|z_{j+1}^{i}-A_{ii}z_{j}^{i}-\sum_{l\in\mathrm{I}\setminus\{i\}}A_{il}\tilde{x}_{j}^{l}\|^{2}_{Q_{i}^{-1}}+\sum_{j=k-N+1}^{k-1}\|y_{j+1}-C(A_{i}^{*}z_{j}^{i}+\tilde{A}_{i}\tilde{x}_{j})\|^{2}_{R^{-1}}
        \nonumber\\
        +&\|z_{k-N+1}^{i}-A_{ii}\hat{x}_{k-N}^{i}-\sum_{l\in\mathrm{I}\setminus\{i\}}A_{il}\tilde{x}_{k-N}^{l}\|^{2}_{Q_{i}^{-1}}+\|y_{k-N+1}-C(A_{i}^{*}\hat{x}_{k-N}^{i}+\tilde{A}_{i}\tilde{x}_{k-N})\|^{2}_{R^{-1}}\nonumber\\
        +&\| y_{k-N}-C_{i}^{*}\hat{x}_{k-N}^{i}-\tilde{C}_{i}\tilde{x}_{k-N}\|^{2}_{R^{-1}}+\|\hat{x}_{k-N}^{i}-\bar{x}_{k-N}^{i}\|^{2}_{P_{i,k-N}^{-1}}+\Phi_{k-1}^{i,*}
    \end{align}
From Lemma \ref{lem:j1+j2}, it holds that
\begin{align}\label{eq:31}
    &\quad\| y_{k-N}-C_{i}^{*}\hat{x}_{k-N}^{i}-\tilde{C}_{i}\tilde{x}_{k-N}\|^{2}_{R^{-1}}+\|\hat{x}_{k-N}^{i}-\bar{x}_{k-N}^{i}\|^{2}_{P_{i,k-N}^{-1}}\nonumber\\
    &= \|\hat{x}_{k-N}^{i}-\breve{x}_{k-N}^{i}\|^{2}_{\breve{P}_{i,k-N}^{-1}}+\text{const}
\end{align}
where
    \begin{align*}
    \breve{P}_{i,k-N}&=P_{i,k-N}-P_{i,k-N}(C_{i}^{*})^{\mathrm{T}}(C_{i}^{*}P_{i,k-N}(C_{i}^{*})^{\mathrm{T}}+R)^{-1}C_{i}^{*}P_{i,k-N}\\
    \breve{x}_{k-N}^{i} &= \bar{x}_{k-N}^{i}+P_{i,k-N}(C_{i}^{*})^{\mathrm{T}}(C_{i}^{*}P_{i,k-N}(C_{i}^{*})^{\mathrm{T}}+R)^{-1}(y_{k-N}-C_{i}^{*}\bar{x}_{k-N}^{i}-\tilde{C}_{i}\tilde{x}_{k-N})
\end{align*}
$const$ represents a constant, which will be defined later. Similarly, it is further obtained:
\begin{align}\label{eq:32}
    &\quad\|\hat{x}_{k-N}^{i}-\breve{x}_{k-N}^{i}\|^{2}_{\breve{P}_{i,k-N}^{-1}}+ \|y_{k-N+1}-C(A_{i}^{*}\hat{x}_{k-N}^{i}+\tilde{A}_{i}\tilde{x}_{k-N})\|^{2}_{R^{-1}}\nonumber\\
    &= \|\hat{x}_{k-N}^{i}-\check{x}_{k-N}^{i}\|^{2}_{\check{P}_{i,k-N}^{-1}}+\text{const}
\end{align}
where
    \begin{align*}
    \check{P}_{i,k-N}&=\breve{P}_{i,k-N}-\breve{P}_{i,k-N}(A_{i}^{*})^{\mathrm{T}}C^{\mathrm{T}}(CA_{i}^{*}\breve{P}_{i,k-N}(A_{i}^{*})^{\mathrm{T}}C^{\mathrm{T}}+R)^{-1}CA_{i}^{*}\breve{P}_{i,k-N}\\
    \check{x}_{k-N}^{i} &= \breve{x}_{k-N}^{i}+\breve{P}_{i,k-N}(A_{i}^{*})^{\mathrm{T}}C^{\mathrm{T}}(CA_{i}^{*}\breve{P}_{i,k-N}(A_{i}^{*})^{\mathrm{T}}C^{\mathrm{T}}+R)^{-1}(y_{k-N+1}-C(A_{i}^{*}\breve{x}_{k-N}^{i}+\tilde{A}_{i}\tilde{x}_{k-N}))
\end{align*}
and
\begin{align}\label{eq:33}
    &\quad\|\hat{x}_{k-N}^{i}-\check{x}_{k-N}^{i}\|^{2}_{\check{P}_{i,k-N}^{-1}}+ \|z_{k-N+1}^{i}-A_{ii}\hat{x}_{k-N}^{i}-\sum_{l\in\mathrm{I}\setminus\{i\}}A_{il}\tilde{x}_{k-N}^{l}\|^{2}_{Q_{i}^{-1}}\nonumber\\
    &= \|\hat{x}_{k-N}^{i}-\sigma\|^{2}_{H^{-1}}+\|z_{k-N+1}^{i}-A_{ii}\check{x}_{k-N}^{i}-\sum_{l\in\mathrm{I}\setminus\{i\}}A_{il}\tilde{x}_{k-N}^{l}\|^{2}_{P_{i,k-N+1}^{-1}}
\end{align}
where
\begin{align*}
    H&=\check{P}_{i,k-N}A_{ii}^{\mathrm{T}}(A_{ii}\check{P}_{i,k-N}A_{ii}^{\mathrm{T}}+Q_{i})^{-1}\\
      \sigma & = \check{x}_{k-N}^{i}+\check{P}_{i,k-N}A_{ii}^{\mathrm{T}}(A_{ii}\check{P}_{i,k-N}A_{ii}^{\mathrm{T}}+Q_{i})^{-1}(z_{k-N+1}^{i}-A_{ii}\check{x}_{k-N}^{i}-\sum_{l\in\mathrm{I}\setminus\{i\}}A_{il}\tilde{x}_{k-N}^{l}) \\
     P_{i,k-N+1}  & =Q_{i}+A_{ii}\check{P}_{i,k-N}A_{ii}^{\mathrm{T}}
\end{align*}
Therefore, by substituting \eqref{eq:31}, \eqref{eq:32}, and \eqref{eq:33} into \eqref{eq:30}, we have
\begin{align}\label{eq:min_u_1}
    &\quad\Phi^{i,u}_{[k-N+1,k]/k}(\{z_{j}^{i}\}_{j=k-N+1}^{k})\nonumber\\
    &=\min_{\hat{x}^{i}_{k-N}} \Big\{\sum_{j=k-N+1}^{k-1}\|z_{j+1}^{i}-A_{ii}z_{j}^{i}-\sum_{l\in\mathrm{I}\setminus\{i\}}A_{il}\tilde{x}_{j}^{l}\|^{2}_{Q_{i}^{-1}}+\sum_{j=k-N+1}^{k-1}\|y_{j+1}-C(A_{i}^{*}z_{j}^{i}+\tilde{A}_{i}\tilde{x}_{j})\|^{2}_{R^{-1}}\nonumber\\
    &\quad+\|\hat{x}_{k-N}^{i}-\sigma\|^{2}_{H^{-1}}+\|z_{k-N+1}^{i}-A_{ii}\check{x}_{k-N}^{i}-\sum_{l\in\mathrm{I}\setminus\{i\}}A_{il}\tilde{x}_{k-N}^{l}\|^{2}_{P_{i,k-N+1}^{-1}}+\text{const}\Big\}
\end{align}
It is noted that the optimal solution of \eqref{eq:min_u_1} is  $\hat{x}_{k-N}^{i}= \sigma$. Consequently, we can obtain that
\begin{align}\label{eq:min_u_2}
&\quad\Phi^{i,u}_{[k-N+1,k]/k}(\{z_{j}^{i}\}_{j=k-N+1}^{k})\nonumber\\
    &=\sum_{j=k-N+1}^{k-1}\|z_{j+1}^{i}-A_{ii}z_{j}^{i}-\sum_{l\in\mathrm{I}\setminus\{i\}}A_{il}\tilde{x}_{j}^{l}\|^{2}_{Q_{i}^{-1}}+\sum_{j=k-N+1}^{k-1}\|y_{j+1}-C(A_{i}^{*}z_{j}^{i}+\tilde{A}_{i}\tilde{x}_{j})\|^{2}_{R^{-1}}\nonumber\\
    &\quad+\|z_{k-N+1}^{i}-A_{ii}\check{x}_{k-N}^{i}-\sum_{l\in\mathrm{I}\setminus\{i\}}A_{il}\tilde{x}_{k-N}^{l}\|^{2}_{P_{i,k-N+1}^{-1}}+\text{const}
\end{align}


Define $Y_{1,i}=A_{ii}\check{x}_{k-N}^{i}+\sum_{l\in\mathrm{I}\setminus\{i\}}A_{il}\tilde{x}_{k-N}^{l}$,
\begin{align*}
    Y_{2,i}&=\left[\begin{array}{ccc}
     \sum_{l\in\mathrm{I}\setminus\{i\}}A_{il}\tilde{x}_{k-N+1}^{l}\\
      \vdots\\
      \sum_{l\in\mathrm{I}\setminus\{i\}}A_{il}\tilde{x}_{k-1}^{l}
    \end{array}\right],
    ~ Y_{3,i}=\left[\begin{array}{ccc}
     y_{k-N+2}-C\tilde{A}_{i}\tilde{x}_{k-N+1}\\
      \vdots\\
     y_{k}-C\tilde{A}_{i}\tilde{x}_{k-1}
    \end{array}\right],
    C_{1,i}=[I_{n_{x^{i}}}, \underbrace{0,\ldots,0}_{N-1}],\\
    C_{2,i}&=\left[\begin{array}{ccccccccc}
        -A_{ii} & I_{n_{x^{i}}} &0 &\ldots&0&0\\
         0&-A_{ii} & I_{n_{x^{i}}} &\ldots&0&0\\
         \vdots&\vdots &\vdots &\ddots&\vdots&\vdots\\
        0&0 & 0 &\ldots&-A_{ii} & I_{n_{x^{i}}}\\
    \end{array}\right],
    C_{3,i} = \left[\begin{array}{ccccccc}
       CA_{i}^{*} & 0  &\ldots&0&0\\
         0&CA_{i}^{*}  &\ldots&0&0\\
         \vdots&\vdots  &\ddots&\vdots&\vdots\\
        0&0  &\ldots&CA_{i}^{*}&0 \\
    \end{array}\right]
\end{align*}
Consequently, \eqref{eq:min_u_2} can be rewritten as
\begin{align}\label{eq:38}
    \Phi^{i,u}_{[k-N+1,k]/k}(\{z_{j}^{i}\}_{j=k-N+1}^{k})&=\|C_{1,i}\{z_{j}^{i}\}_{j=k-N+1}^{k}-Y_{1,i}\|^{2}_{P_{i,k-N+1}^{-1}}+\|C_{2,i}\{z_{j}^{i}\}_{j=k-N+1}^{k}-Y_{2,i}\|^{2}_{\mathbf{Q}_{i}^{-1}}\nonumber\\
    &\quad+\|C_{3,i}\{z_{j}^{i}\}_{j=k-N+1}^{k}-Y_{3,i}\|^{2}_{\mathbf{R}^{-1}}+\text{const}\nonumber\\
    &=\|C_{4,i}\{z_{j}^{i}\}_{j=k-N+1}^{k}-Y_{4,i}\|^{2}_{(\tilde{H}_{k}^{i})^{-1}}+\text{const}
\end{align}
where $\mathbf{Q}_{i}=\mathrm{diag}\{\underbrace{Q_{i},\ldots,Q_{i}}_{N-1}\}$; $\mathbf{R}=\mathrm{diag}\{\underbrace{R,\ldots,R}_{N-1}\}$;
$Y_{4,i}=\mathrm{col}\{Y_{1,i},Y_{2,i},Y_{3,i}\}$; $C_{4,i}=\mathrm{col}\{C_{1,i},C_{2,i},C_{3,i}\}$; $\tilde{H}_{k}^{i}=\mathrm{diag}\{P_{i,k-N+1},\mathbf{Q}_{i},\mathbf{R}\}$.
According to Lemma \ref{lem3}, \eqref{eq:38} is equivalent to the equation
\begin{equation}\label{eq:min_u_Y}
    \Phi^{i,u}_{[k-N+1,k]/k}(\{z_{j}^{i}\}_{j=k-N+1}^{k})=\|\{z_{j}^{i}\}_{j=k-N+1}^{k}-Y_{i}\|^{2}_{H_{k}^{i}}+\text{const}
\end{equation}
where
\begin{subequations}
    \begin{align}
    H_{k}^{i}&=C_{4,i}^{\mathrm{T}}(\tilde{H}_{k}^{i})^{-1}C_{4,i}\label{eq:H}\\
    Y_{i} &=(C_{4,i}^{\mathrm{T}}(\tilde{H}_{k}^{i})^{-1}C_{4,i})^{-1}C_{4,i}^{\mathrm{T}}(\tilde{H}_{k}^{i})^{-1}Y_{4,i}
\end{align}
\end{subequations}
By optimality, $z_{j}^{i}=\hat{x}_{j}^{i,u}$, where $j=k-N+1,\ldots,k$, are the global minimizers of transit cost $\Phi^{i,u}_{[k-N+1,k]/k}(\{z_{j}^{i}\}_{j=k-N+1}^{k})$ in \eqref{eq:transit_u}, and the corresponding global minimum of the transit cost in \eqref{eq:transit_u} is $\Phi_{k}^{i,u}$. Therefore, we have that $\{\hat{x}_{j}^{i,u}\}_{j=k-N+1}^{k}=Y_{i}$, and the constant term in \eqref{eq:min_u_Y} is $\Phi_{k}^{i,u}$. $\square$
\end{proof}

\begin{lemma}\label{Lemma:14}(\cite{farina2024moving})
    Let
     $\{\hat{x}_{j}^{i}\}_{j=k-N}^{k}$ be the solution of the proposed DMHE problem in \eqref{dmhe_revised} for the subsystem $i$.
Then it holds
\begin{equation*}
   \Phi^{i,\mathrm{MHE}}_{[k-N+1,k]/k}(\{z_{j}^{i}\}_{j=k-N+1}^{k})\geq\|\{z_{j}^{i}\}_{j=k-N+1}^{k}-\{\hat{x}_{j}^{i}\}_{j=k-N+1}^{k}\|^{2}_{H_{k}^{i}}+\Phi^{i,*}_{k}
\end{equation*}
\end{lemma}

Before proceeding further, we introduce a matrix and an assumption that will be used to prove the stability of the proposed DMHE method in \eqref{dmhe_revised}.
\begin{align}\label{eq:W}
    W_{k}^{i}&=\mathrm{diag}\{P_{l,k-N}^{-1}+C^{\mathrm{T}}R^{-1}C+n\sum_{i\in\mathbb{I}}(A_{r,il}^{\mathrm{T}}Q_{i}^{-1}A_{r,il}+\tilde{A}_{i,[:,i]}^{\mathrm{T}}C^{\mathrm{T}}R^{-1}C\tilde{A}_{i,[:,i]}),\ldots,\nonumber\\
    &\quad n\sum_{i\in\mathbb{I}}(A_{r,il}^{\mathrm{T}}Q_{i}^{-1}A_{r,il}+\tilde{A}_{i,[:,i]}^{\mathrm{T}}C^{\mathrm{T}}R^{-1}C\tilde{A}_{i,[:,i]})\}
\end{align}

\begin{assumption}\label{assum:matrix}
The matrices $H_{k}^{i}$ and $W_{k}^{i}$ satisfy
\begin{equation*}
    W_{k}^{i}\leq H_{k}^{i},~ \forall i\in\mathbb{I}
\end{equation*}
where $H_{k}^{i}$ and $W_{k}^{i}$ are defined in \eqref{eq:H} and \eqref{eq:W}, respectively.
\end{assumption}

\begin{proposition}\label{lem:6}
For the proposed DMHE approach in \eqref{dmhe_revised}, if Assumption \ref{assum:matrix} holds, then
    \begin{equation}\label{eq:lem6}
        \sum_{j=k-N}^{k-1}\|\hat{w}_{j}\|^{2}_{Q_{i}^{-1}}+\sum_{j=k-N}^{k}\|\hat{V}_{j}\|^{2}_{\mathbf{R}^{-1}}\mathop{\longrightarrow}\limits^{k\rightarrow\infty}0
    \end{equation}
\end{proposition}
\begin{proof}
    Based on \eqref{eq:collective_24} and \eqref{eq:collective_25}, we can obtain the following
    \begin{equation}\label{eq:increasing}
        \Phi^{*}_{k} -\Phi^{*}_{k-1}\geq\sum_{j=k-N}^{k-1}\|\hat{w}_{j}\|^{2}_{Q_{i}^{-1}}+\sum_{j=k-N}^{k}\|\hat{V}_{j}\|^{2}_{\mathbf{R}^{-1}}\geq0
    \end{equation}
    Therefore, the sequence $\Phi^{*}_{k}$ is increasing.
According to optimality $\Phi^{*}_{k}$, it follows that
    \begin{align}\label{eq:iter_1}
        \Phi^{*}_{k}&\leq\Phi_{[k-N+1,k]/k}^{\mathrm{MHE}}(\{x_{j}\}_{j=k-N+1}^{k})
    \end{align}
    where $x_{j}$, $j=k-N+1,\ldots,k$, is the actual state generated by \eqref{cmodel} without process disturbances and measurement noise.
From \eqref{collective_estimation}, by choosing $\hat{x}_{k-N}=x_{k-N}$ and  $\hat{w}_{j}=A_{r}(x_{j}-\tilde{x}_{j})$ for $j=k-N,\ldots,k-1$, the trajectory of $\hat{x}_{j}=x_{j}$, for $j=k-N+1,\ldots,k$, can be generated. Then, we have
    \begin{equation}\normalsize
\hat{V}_{j} =\left\{ \begin{array}{l}
\mathbf{C}\tilde{A}(x_{j-1}-\tilde{x}_{j-1}),~~~~~{\text{for}}~ j=k-N+1,\ldots,k~~\\[0.3em]
C^{*}(x_{j}-\tilde{x}_{j}),~~~~~~~~~~~{\text{for}}~ j=k-N~~
\end{array} \right.
\end{equation}
Therefore, by optimality, it holds that
    \begin{align}\label{eq:49}
         \Phi_{[k-N+1,k]/k}^{\mathrm{MHE}}(\{x_{j}\}_{j=k-N+1}^{k})
        &\leq\sum_{j=k-N}^{k-1}\|A_{r}(x_{j}-\tilde{x}_{j})\|^{2}_{Q^{-1}}+\sum_{j=k-N}^{k-1}\|\mathbf{C}\tilde{A}(x_{j}-\tilde{x}_{j})\|^{2}_{\mathbf{R}^{-1}}\\
        &\quad+\|C^{*}(x_{k-N}-\tilde{x}_{k-N})\|^{2}_{\mathbf{R}^{-1}}+\|x_{k-N}-\bar{x}_{k-N}\|^{2}_{P_{k-N}^{-1}}+\Phi^{*}_{k-1}\nonumber
    \end{align}

Next, our objective is to prove that $\Phi_{[k-N+1,k]/k}^{\mathrm{MHE}}(\{x_{j}\}_{j=k-N+1}^{k})\leq\Phi_{[k-N,k-1]/k-1}^{\mathrm{MHE}}(\{x_{j}\}_{j=k-N}^{k-1})$. To achieve this, we analyze each term on the right-hand-side of \eqref{eq:49}.
Specifically, the first term on the right-hand-side of \eqref{eq:49} satisfies
\begin{align}\label{eq:44_1}
    \sum_{j=k-N}^{k-1}\|A_{r}(x_{j}-\tilde{x}_{j})\|^{2}_{Q^{-1}}&=\sum_{i\in\mathbb{I}}\sum_{j=k-N}^{k-1}\|\sum_{l\in\mathbb{I}}A_{r,il}(x_{j}^{l}-\tilde{x}_{j}^{l})\|^{2}_{Q_{i}^{-1}}\nonumber\\
    &\leq \sum_{i\in\mathbb{I}}\sum_{j=k-N}^{k-1}\sum_{l\in\mathbb{I}}n\|A_{r,il}(x_{j}^{l}-\tilde{x}_{j}^{l})\|^{2}_{Q_{i}^{-1}}\nonumber\\
    &= \sum_{l\in\mathbb{I}}\sum_{j=k-N}^{k-1}\|x_{j}^{l}-\tilde{x}_{j}^{l}\|^{2}_{n\sum_{i\in\mathbb{I}}A_{r,il}^{\mathrm{T}}Q_{i}^{-1}A_{r,il}}
\end{align}
where $A_{r,il}$ represents the block matrix in the $i$th row and the $l$th column of the matrix $A_{r}$.
Similarly, by analyzing the remaining terms on the right-hand-side of \eqref{eq:49}, one can obtain
\begin{subequations}\label{eq:44_2}
    \begin{align}
    \sum_{j=k-N}^{k-1}\|\mathbf{C}\tilde{A}(x_{j}-\tilde{x}_{j})\|^{2}_{\mathbf{R}^{-1}}
    &\leq \sum_{l\in\mathbb{I}}\sum_{j=k-N}^{k-1}\|x_{j}^{l}-\tilde{x}_{j}^{l}\|^{2}_{n\sum_{i\in\mathbb{I}}\tilde{A}_{i,[:,i]}^{\mathrm{T}}C^{\mathrm{T}}R^{-1}C\tilde{A}_{i,[:,i]}}\\
    \|\mathbf{C}(x_{k-N}-\tilde{x}_{k-N})\|^{2}_{\mathbf{R}^{-1}}
    &\leq \sum_{l\in\mathbb{I}}\|x_{k-N}^{l}-\tilde{x}_{k-N}^{l}\|^{2}_{C^{\mathrm{T}}R^{-1}C}\\
    \|x_{k-N}-\bar{x}_{k-N}\|^{2}_{P_{k-N}^{-1}}
    &\leq \sum_{l\in\mathbb{I}}\|x_{k-N}^{l}-\tilde{x}_{k-N}^{l}\|^{2}_{P_{l,k-N}^{-1}}
\end{align}
\end{subequations}
where $\tilde{A}_{i,[:,i]}$ is composed of the columns of $\tilde{A}_{i}$ with respect to subsystem state $x^{i}$. Then, substituting \eqref{eq:44_1} and \eqref{eq:44_2} into \eqref{eq:49} yields
\begin{align*}
    \Phi_{[k-N+1,k]/k}^{\mathrm{MHE}}(\{x_{j}\}_{j=k-N+1}^{k})&\leq \sum_{l\in\mathbb{I}}\|x_{k-N}^{l}-\tilde{x}_{k-N}^{l}\|^{2}_{P_{l,k-N}^{-1}+C^{\mathrm{T}}R^{-1}C}\nonumber\\
    &\quad+\sum_{l\in\mathbb{I}}\sum_{j=k-N}^{k-1}\|x_{j}^{l}-\tilde{x}_{j}^{l}\|^{2}_{n\sum_{i\in\mathbb{I}}(A_{r,il}^{\mathrm{T}}Q_{i}^{-1}A_{r,il}+\tilde{A}_{i,[:,i]}^{\mathrm{T}}C^{\mathrm{T}}R^{-1}C\tilde{A}_{i,[:,i]})}+\Phi^{*}_{k-1}\nonumber\\
    &\leq\sum_{l\in\mathbb{I}}\big(\|\{x_{j}^{l}\}_{j=k-N}^{k-1}-\{\tilde{x}_{j}^{l}\}_{j=k-N}^{k-1}\|^{2}_{W_{k}^{i}}+\Phi^{i,*}_{k-1}\big)
\end{align*}
where $W_{k}^{i}$ is defined in \eqref{eq:W}.
Considering Assumption \ref{assum:matrix} and Lemma \ref{Lemma:14}, one can obtain
\begin{align}\label{eq:iter_2}
\Phi_{[k-N+1,k]/k}^{\mathrm{MHE}}(\{x_{j}\}_{j=k-N+1}^{k})
&\leq\sum_{l\in\mathbb{I}}\Big(\|\{x_{j}^{l}\}_{j=k-N}^{k-1}-\{\hat{x}_{j}^{l}\}_{j=k-N}^{k-1}\|^{2}_{H_{k}^{i}}+\Phi^{i,*}_{k-1}\Big)\nonumber\\
&\leq\sum_{l\in\mathbb{I}}\Phi^{i,\mathrm{MHE}}_{[k-N,k-1]/k-1}(\{z_{j}^{i}\}_{j=k-N}^{k-1})\nonumber\\
&=\Phi_{[k-N,k-1]/k-1}^{\mathrm{MHE}}(\{z_{j}\}_{j=k-N}^{k-1})
\end{align}
From \eqref{eq:iter_1} and \eqref{eq:iter_2}, we can iterate this procedure and obtain
that
\begin{equation}\label{eq:bounded}
    \Phi^{*}_{k}\leq\Phi_{[k-N+1,k]/k}^{\mathrm{MHE}}(\{x_{j}\}_{j=k-N+1}^{k})\leq\Phi_{[k-N,k-1]/k-1}^{\mathrm{MHE}}(\{z_{j}\}_{j=k-N}^{k-1})\leq\ldots\leq\|x_{0}-\bar{x}_{0}\|^{2}_{P_{0}^{-1}}
\end{equation}
Considering \eqref{eq:increasing} and \eqref{eq:bounded}, the sequence of $\Phi^{*}_{k}$ converges as it is increasing and bounded. Consequently, \eqref{eq:lem6} is proven. $\square$
\end{proof}

We further define the estimation error at sampling instant $j$ calculated at sampling instant $k$ as $e_{j|k}=x_{j}-\hat{x}_{j|k}=x_{j}-\hat{x}_{j}$. Then, the estimation error of sampling instant $j$ calculated at the previous sampling instant $k-1$ is denoted by $e_{j|k-1}=x_{j}-\hat{x}_{j|k-1}=x_{j}-\tilde{x}_{j}$.
\begin{theorem}\label{theorem}
    If Assumption \ref{assum:matrix} holds, then there exists a sequence $\alpha_{k}$, $\alpha_{k}\mathop{\longrightarrow}\limits^{k\rightarrow\infty}0$, such that the sequence of estimation error within the estimation window $E_{k}=\mathrm{col}\{e_{k-N+1|k},\ldots,e_{k|k}\}$ for the entire system in \eqref{cmodel} generated by the proposed DMHE in \eqref{dmhe_revised} is described by
\begin{align*}
    E_{k}
        &=(M_{2}-M_{1}(O^{\mathrm{T}}O)^{-1}O^{\mathrm{T}}\Gamma)E_{k-1}+\alpha_{k}
\end{align*}
Additionally, the estimation error $E_{k}$ converges, if the spectral radius of matrix $M_{2}-M_{1}(O^{\mathrm{T}}O)^{-1}O^{\mathrm{T}}\Gamma$ satisfies
\begin{equation*}
    \rho(M_{2}-M_{1}(O^{\mathrm{T}}O)^{-1}O^{\mathrm{T}}\Gamma)<1
\end{equation*}
where
\begin{align}\label{eq:notation}
    O&=\left[\begin{array}{ccccccccc}
        \mathbf{C}A^{*}  \\
         \vdots\\
         \mathbf{C}A^{*}A_{d}^{N-1}
    \end{array}\right],~
    \Gamma=\left[\begin{array}{ccccccccc}
        \mathbf{C}\tilde{A}&0&\cdots&0  \\
         \mathbf{C}A^{*}A_{r}&\mathbf{C}\tilde{A}&\cdots&0\\
         \vdots& \vdots& \vdots& \vdots\\
         \mathbf{C}A^{*}A_{d}^{N-2}A_{r}&\mathbf{C}A^{*}A_{d}^{N-1}A_{r}&\cdots&\mathbf{C}\tilde{A}\\
    \end{array}\right]\nonumber\\
    M_{1}&=\left[\begin{array}{ccccccccc}
        A_{d}  \\
         \vdots\\
         A_{d}^{N}
    \end{array}\right],~
    M_{2}=\left[\begin{array}{ccccccccc}
        A_{r}&0&\cdots&0  \\
        A_{d}A_{r}&A_{r}&\cdots&0\\
         \vdots& \vdots& \vdots& \vdots\\
         A_{d}^{N-1}A_{r}&A_{d}^{N-2}A_{r}&\cdots&A_{r}\\
    \end{array}\right]
\end{align}
\end{theorem}
\begin{proof}
In the noise-free setting (i.e., $w_{k}=0$ and $v_{k}=0$, $\forall k$), the actual state satisfies
\begin{subequations}\label{collective_actual}
    \begin{align}
    x_{j+1}&=A_{d}x_{j}+A_{r}x_{j}\label{eq:50a}\\
    Y_{j+1}&=C(A^{*}x_{j}+\tilde{A}x_{j})\label{eq:50b}
\end{align}
\end{subequations}
Considering \eqref{eq:23a} and \eqref{eq:50a}, one can obtain
\begin{subequations}\label{eq:x_xhat}
    \begin{align}
  x_{j}&=A_{d}^{j-(k-N)}x_{k-N}+\sum_{l=1}^{j-(k-N)}A_{d}^{l-1}A_{r}x_{j-l}\\
  \hat{x}_{j}&=A_{d}^{j-(k-N)}\hat{x}_{k-N}+\sum_{l=1}^{j-(k-N)}A_{d}^{l-1}A_{r}\tilde{x}_{j-l}+\sum_{l=1}^{j-(k-N)}A_{d}^{l-1}\hat{w}_{j-l}
\end{align}
\end{subequations}
From \eqref{eq:x_xhat}, it it is further derived that
\begin{align}\label{eq:53}
    e_{j|k}=A_{d}^{j-(k-N)}e_{k-N|k}+\sum_{l=1}^{j-(k-N)}A_{d}^{l-1}A_{r}e_{j-l|k-1}-\sum_{l=1}^{j-(k-N)}A_{d}^{l-1}\hat{w}_{j-l}
\end{align}
By taking into account \eqref{eq:23c}, \eqref{eq:50b}, and \eqref{eq:53}, we have
\begin{align}\label{eq:54}
    \sum_{j=k-N}^{k-1}\|\hat{V}_{j+1}\|&=\sum_{j=k-N}^{k-1}\|\mathbf{C}A^{*}e_{j|k}+\mathbf{C}\tilde{A}e_{j|k-1}\|\nonumber\\
    &\geq\sum_{j=k-N}^{k-1}\|\mathbf{C}A^{*}A_{d}^{j-(k-N)}e_{k-N|k}+\mathbf{C}A^{*}\sum_{l=1}^{j-(k-N)}A_{d}^{l-1}A_{r}e_{j-l|k-1}+\mathbf{C}\tilde{A}e_{j|k-1}\|\nonumber\\
    &\quad-\|\mathbf{C}\|\|A^{*}\|\sum_{j=k-N}^{k-1}\sum_{l=1}^{j-(k-N)}\|A_{d}^{l-1}\|\|\hat{w}_{j-l}\|
\end{align}
Therefore, \eqref{eq:54} is equivalent to
\begin{align*}
&\quad\sum_{j=k-N}^{k-1}\|\mathbf{C}A^{*}A_{d}^{j-(k-N)}e_{k-N|k}+\mathbf{C}A^{*}\sum_{l=1}^{j-(k-N)}A_{d}^{l-1}A_{r}e_{j-l|k-1}+\mathbf{C}\tilde{A}e_{j|k-1}\|\nonumber\\
&\leq\sum_{j=k-N}^{k-1}\|\hat{V}_{j+1}\|+\|\mathbf{C}\|\|A^{*}\|\sum_{j=k-N}^{k-1}\sum_{l=1}^{j-(k-N)}\|A_{d}^{l-1}\|\|\hat{w}_{j-l}\|
\end{align*}
Based on Proposition \ref{lem:6}, it is obtained that
\begin{align*}
    &\quad\sum_{j=k-N}^{k-1}\|\mathbf{C}A^{*}A_{d}^{j-(k-N)}e_{k-N|k}+\mathbf{C}A^{*}\sum_{l=1}^{j-(k-N)}A_{d}^{l-1}A_{r}e_{j-l|k-1}+\mathbf{C}\tilde{A}e_{j|k-1}\|\nonumber\\
    &=\|Oe_{k-N|k}+\Gamma E_{k-1}\|\mathop{\longrightarrow}\limits^{k\rightarrow\infty}0
\end{align*}
where $O$ and $\Gamma$ are defined in \eqref{eq:notation}.
Let $\alpha_{k}^{1}$ and $\alpha_{k}^{2}$ denote asymptotically vanishing variables, i.e., $\|\alpha_{k}^{j}\|\mathop{\longrightarrow}\limits^{k\rightarrow\infty}0$, $j=1,2$.
Then, it holds
\begin{equation}\label{eq:alpha1}
   Oe_{k-N|k}+\Gamma E_{k-1}=\alpha_{k}^{1}
\end{equation}
By concatenating \eqref{eq:53} for $j=k-N+1,\ldots,k$, we can obtain
\begin{equation}\label{eq:alpha2}
  E_{k}=M_{1}e_{k-N|k}+M_{2}E_{k-1}+\alpha_{k}^{2}
\end{equation}
where $M_1$ and $M_2$ are defined in \eqref{eq:notation}.
From \eqref{eq:alpha1} and \eqref{eq:alpha2}, it holds
\begin{align*}
        E_{k}&=(M_{2}-M_{1}(O^{\mathrm{T}}O)^{-1}O^{\mathrm{T}}\Gamma)E_{k-1}+\alpha_{k}
\end{align*}
where $\alpha_{k}=\alpha_{k}^{2}-M_{1}(O^{\mathrm{T}}O)^{-1}O^{\mathrm{T}}\alpha_{k}^{1}$, which satisfies $\alpha_{k}\mathop{\longrightarrow}\limits^{k\rightarrow\infty}0$. Additionally, the estimation error $E_{k}$ converges to zero when $\rho(M_{2}-M_{1}(O^{\mathrm{T}}O)^{-1}O^{\mathrm{T}}\Gamma)<1$. $\square$
\end{proof}

\section{Distributed moving horizon estimation for nonlinear systems}
In this section, the arrival cost design obtained for linear systems is extended to approximate the arrival costs of the local estimators in the nonlinear constrained context.

\subsection{Arrival cost approximation for nonlinear systems} \label{sec:arrival cost}
We extend the arrival cost design for linear unconstrainted systems, as depicted in \eqref{eq:upd:arrival cost_1}, to approximate the arrival costs for nonlinear systems. %
The nonlinear subsystem model in \eqref{model_nonlinear_centralized} is used as the model basis for each local estimator.
Through successively linearizing the subsystem model in \eqref{model_nonlinear_centralized} at each sampling instant $k$, we obtain an approximation of the arrival cost for the $i$th estimator.
The linearization is performed as follows:
 \begin{equation}\label{eq:linearization}
    C_{k} =\frac{\partial h(x_{k})}{\partial \hat{x}_{k}},\quad\quad\quad A_{k} = \frac{\partial f(x_{k}, X_{k})}{\partial \hat{x}_{k}}.
  \end{equation}
Then, the corresponding matrices $A_{il,k}$, and $A_{[:,i],k}$, $\forall i,l\in\mathbb{I}$, can be derived from $A_{k}$ in \eqref{eq:linearization}. These matrices are utilized to update the arrival cost $V_{k}^{i,o}$ of the proposed constrained nonlinear DMHE.
Specifically, for each subsystem $i$, $i\in\mathbb{I}$, the expression of the proposed arrival cost design for nonlinear systems is obtained from \eqref{eq:upd:arrival cost_1} and \eqref{eq:upd:arrival cost_2} by replacing $A_{il}$ with $A_{il,k}$, $A_{[:,i]}$ with $A_{[:,i],k}$, $\forall i,l\in\mathbb{I}$, and $C$ with $C_{k}$.

\subsection{Formulation of nonlinear MHE-based estimators}
Based on the arrival cost approximation for nonlinear constrained systems outlined in Section \ref{sec:arrival cost}, at each sampling instant $k$, the local estimator of the DMHE algorithm for the nonlinear system in \eqref{model_nonlinear_centralized} is as follows:
\normalsize
\begin{subequations}\label{eq:nmhe}
  \begin{align}
&\min_{\{\hat{x}^{i}_{j}\}_{j=k-N}^{k}}\sum_{j=k-N}^{k-1}\|\hat{w}^{i}_{j}\|^2_{Q_{i}^{-1}}+\sum_{j=k-N}^{k}\|\hat{v}_{j}^{[i]}\|^2_{R^{-1}}
+\|\hat{x}^{i}_{k-N}-\bar{x}^{i}_{k-N}\|_{P_{i,k-N}^{-1}}^{2}\\
 &\quad~\,\,~\text{s.t.}~ \hat{x}_{j+1}^{i}=f_{i}(\hat{x}_{j}^{i}, \tilde{X}_{j}^{i})+\hat{w}_{j}^{i}\\
&\quad\,\,~\quad\quad\quad\, y_{0}= h(\hat{x}_{0}^{[i]})+\hat{v}_{0}^{[i]}\\
&\quad\,\,~\quad\quad\quad\, y_{j+1}= h(f(\hat{x}_{j}^{[i]}))+\hat{v}_{j}^{[i]},~ j=0,\ldots,k-1\\
&\quad\,\,~\quad\quad\quad\, \hat{x}_{j}^{i}\in\mathbb{X}^{i}, ~\hat{w}_{j}^{i}\in\mathbb{W}^{i}
 \end{align}
\end{subequations}
\normalsize
{\color{black} In (\ref{eq:nmhe}), $\hat{x}_{j}^{[i]}=\mathrm{col}\{\tilde{x}_{j}^{1},\ldots,\hat{x}_{j}^{i}, \ldots,\tilde{x}_{j}^{n}$\}, where $\hat{x}^{i}_{j}$ represents the state estimate of subsystem $i$ and serves as the decision variable of the optimization problem in \eqref{eq:nmhe}, and $\tilde{x}_{j}^{l}$ is determined based on the estimate of each interconnected subsystem $l$, $l\in\mathbb{I}\setminus\{i\}$, generated at the previous sampling instant $k-1$; $\tilde{X}_{j}^{i}$ concatenates all the $\tilde{x}_{j}^{l}$ of the interconnected subsystems $l$, $l\in\mathbb{I}\setminus\{i\}$;}
$\mathbb{X}^{i}$ and $\mathbb{W}^{i}$ are two compact sets that contain $\hat{x}_{j}^{i}$ and $\hat{w}_{j}^{i}$, respectively.
When solving the MHE-based optimization problem for the $i$th subsystem, only the state estimates associated with subsystem $i$, i.e., $\hat{x}^{i}_j$, $j=k-N,\ldots,k$, are treated as decision variables. The state estimates $\tilde{x}^{l}_{j}$ of the interconnected subsystems $l$, $l\in\mathbb{I}\setminus\{i\}$, are considered as known inputs to the $i$th estimator.

Algorithm \ref{alg1} outlines the implementation steps for the proposed DMHE approach for the nonlinear system in \eqref{model_nonlinear_centralized} with a recursive update of arrival costs of the local estimators, which can be followed to generate the optimal state estimates $\hat{x}_{k}^{i}$ for the $i$th subsystem, $i\in\mathbb{I}$.

\begin{algorithm}[t]
\caption{Key steps for the implementation of the proposed DMHE method}\vspace{3mm}
\label{alg1}
At each sampling time $k\geq N+1$, the MHE-based estimator for the $i$th subsystem, $i\in\mathbb{I}$, carry out the following steps:
\begin{enumerate}
    \item[1.] Receive measured outputs $\{y\}_{k-N}^{k}$, and optimal estimates $\{\tilde{x}^{l}_{j}\}_{j=k-N-1}^{k-1}$ of $l$th subsystem, $l\in\mathbb{I}\setminus\{i\}$, obtained at the previous sampling instant $k-1$ from each estimator $l$, $l\in\mathbb{I}\setminus\{i\}$.
     \item[2.] Compute the open-loop state prediction $\bar{x}_{k-N}^{i}$ following \eqref{eq:20_1}, \eqref{eq:20_2}, \eqref{eq:20_3}, and \eqref{eq:linearization}.
     \item[3.] Compute the weighting matrix $P_{i,k}$ via \eqref{eq:20_4}, \eqref{eq:20_5}, \eqref{eq:20_6}, and \eqref{eq:linearization}.
    \item[4.] Solve \eqref{eq:nmhe} to generate optimal state estimates (i.e., $\{\hat{x}_{j}^{i}\}_{j=k-N}^{k}$).
     \item[5.] Set $k=k+1$. Go to step 1.

     \end{enumerate}
\end{algorithm}

\section{Application to a reactor-separator process}
\begin{figure}[tttt]
 \centering
 \includegraphics[width=0.8\textwidth]{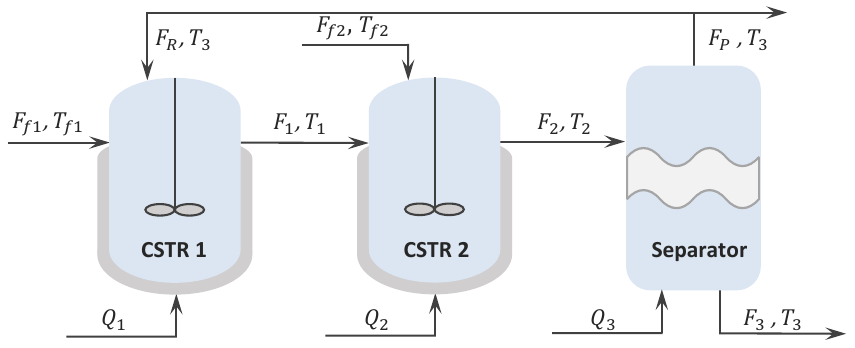}
 \caption{A schematic of the reactor-separator process.}\label{CSTR}
\end{figure}
\subsection{Process description}\label{sim:description}
In this section, we consider a reactor-separator chemical process that consists of two continuous stirred tank reactors (CSTRs) and one flash tank separator. Based on the physical topology presented in Figure \ref{CSTR}, we partition this process into three subsystems, with each subsystem accounting for one vessel.

This chemical process involves two reactions: the first reaction converts reactant $A$ into desired product $B$; the second reaction converts $B$ into side product $C$.
The system states include the mass fractions of reactant $A$ (denoted by $x_{Ai}$, $i=1,2,3$), the mass fractions of product $B$ (denoted by $x_{Bi}$,  $i=1,2,3$), and the temperatures in three vessels (denoted by $T_{i}$, $i=1,2,3$). Among these states, only temperatures $T_{i}$ in the three vessels can be measured online. The details of the first-principles nonlinear dynamic model and a more comprehensive description of this chemical process can be found in \cite{zhang2013distributed}. The objective is to implement the proposed DMHE approach, where the arrival costs for the local estimators are updated using a recursive method for estimating the nine system states based on the measured outputs $T_{i}$, $i=1,2,3$.

\begin{table}[tttt]
  \centering
    \caption{The initial state $x_{0}$ and the initial guess $\bar{x}_{0}$ for the chemical process.}\label{tbl:initial_states}
  \begin{tabular}{cccccccccc}
   \toprule
   & $x_{A1}$ & $x_{B1}$ & $T_{1}$ $(\mathrm{K})$ & $x_{A2}$ &$x_{B2}$& $T_{2}$ $(\mathrm{K})$ & $x_{A3}$ & $x_{B3}$ & $T_{3}$ $(\mathrm{K})$ \\
   \midrule
   $x_{0}$ & 0.1939 & 0.7404 & 528.3482 & 0.2162  & 0.7190& 520.0649 & 0.0716 & 0.7373 & 522.3765\\
    $\bar{x}_{0}$ & 0.2521 & 0.9625 & 686.8525 & 0.2810 & 0.9346& 676.0844 & 0.0931 & 0.9585 & 679.0894\\
   \bottomrule
\end{tabular}
\end{table}

In the simulations, the heat exchange rates considered in this process are $Q_1 = (2.9+1.74$ $\sin(0.06\pi t))\times 10^6~\mathrm{kJ/h}$, $Q_2=(1+0.6\sin(0.06\pi t))$ $\times 10^6~\mathrm{kJ/h}$, and $Q_3=(2.9+1.74\sin(0.06\pi t))\times 10^6~\mathrm{kJ/h}$. The initial state $x_{0}$ that is utilized to generate the actual state trajectories of this process is presented in Table \ref{tbl:initial_states}.
All the states are scaled to ensure equal importance is assigned to states of different magnitudes.
Unknown process disturbances and measurement noise are generated following a zero-mean Gaussian distribution with a standard deviation of $0.01$ for process disturbances and $0.05$ for measurement noise, which are further added to the states and output measurements in the scaled coordinate, respectively.

\subsection{Simulation results}
\begin{figure}[tttt]
  \centering
  \includegraphics[width=0.9\textwidth]{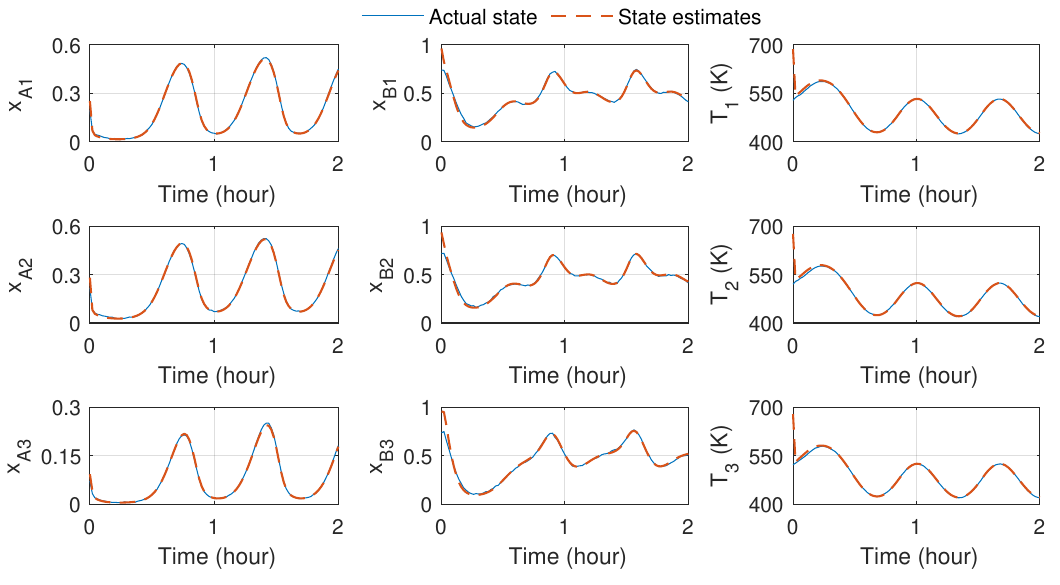}
  \caption{Trajectories of the actual system state and the state estimates provided by the proposed DMHE algorithm for three vessels.}\label{estimation}
\end{figure}

\begin{table}[tttt]
  \centering
    \caption{Comparison of RMSEs.}\label{error}
  \begin{tabular}{cccccccccc}
   \toprule
   & Proposed DMHE & DMHE-1 & DMHE-2 \\
   \midrule
   RMSE & 0.1384 & 0.1425 & 0.2595 \\
   \bottomrule
\end{tabular}
\end{table}
The estimation window is $N=4$. {\color{black}The initial guess for the proposed DMHE is picked as $\bar{x}_{0}=1.3\times x_{0}$, as presented in Table \ref{tbl:initial_states}.} The initial weighting matrices $P_{i,0}$, $Q_{i}$, and $R_{i}$, $i\in\mathbb{I}$, are chosen as $P_{i,0}=0.001\times I_{3}$, $Q_{i}=0.01\times I_{3}$, and $R=0.05\times I_{3}$.  We impose constraints
on the estimates of $x_{Ai}$ and $x_{Bi}$, $i=1,2,3$, generated by the proposed DMHE such that they stay within the range of $[0,1]$, while the estimates of temperatures $T_{i}$, $i=1,2,3$, are made positive.

\begin{figure}[tttt]
  \centering
  \includegraphics[width=0.9\textwidth]{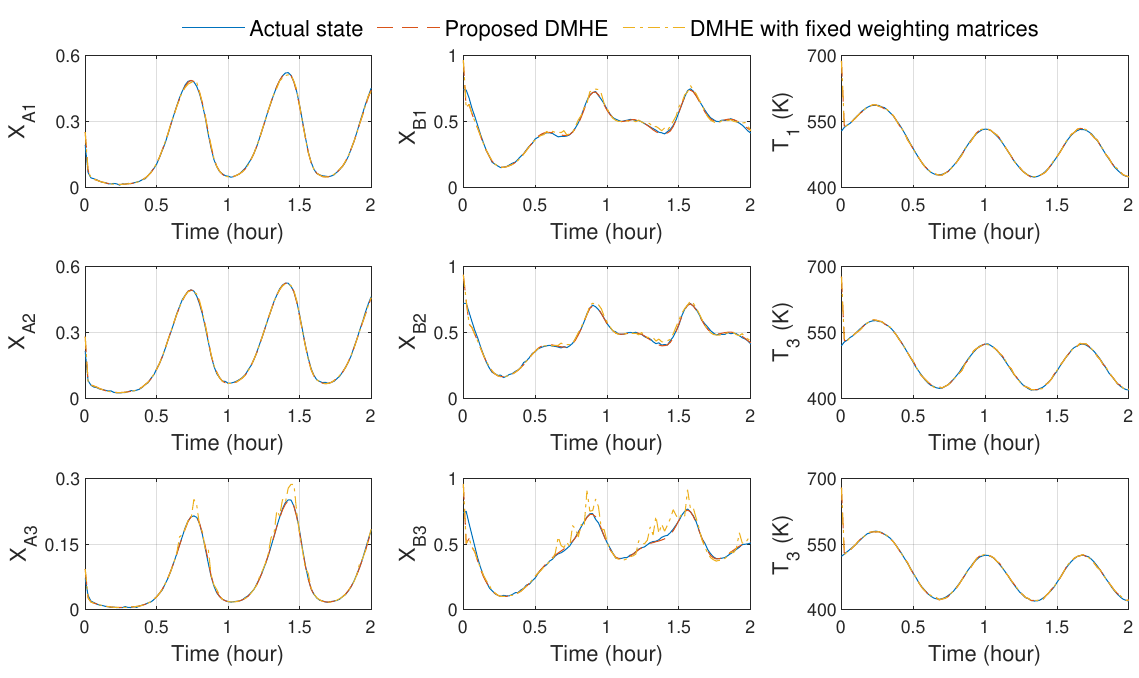}
  \caption{Trajectories of the actual system state and the state estimates provided by the proposed DMHE algorithm and the DMHE-3 approach for three vessels.}\label{compare}
\end{figure}

The trajectories of the state estimates given by the proposed DMHE algorithm and the actual states are presented in Figure \ref{estimation}.
{\color{black}The proposed DMHE approach provides accurate estimates of the ground truth of all the process states, which demonstrates the robustness of the proposed DMHE approach against unknown disturbances.}
Additionally, we evaluate and compare the estimation performance of the proposed DMHE approach with two DMHE algorithms of which each local estimator only uses the sensor measurements of the corresponding subsystems: 1) DMHE-1, where the arrival cost is constructed as a weighted squared error between the state estimate and the \emph{a priori} state prediction, with the weighting matrix determined by a constant matrix; 2) DMHE-2, where the arrival cost is not considered in DMHE algorithm.
The constant weighting matrices $P_{i}$, $Q_{i}$, and $R$, $i=1,2,3$, of DMHE-1, are chosen the same as the initial weighting matrices of the proposed DMHE method. The root mean squared errors (RMSE) for the three DMHE algorithms in the scaled coordinate are shown in Table \ref{error}. The proposed DMHE algorithm provides more accurate estimates than the other two methods.

\begin{figure}[tttt]
  \centering
  \includegraphics[width=0.8\textwidth]{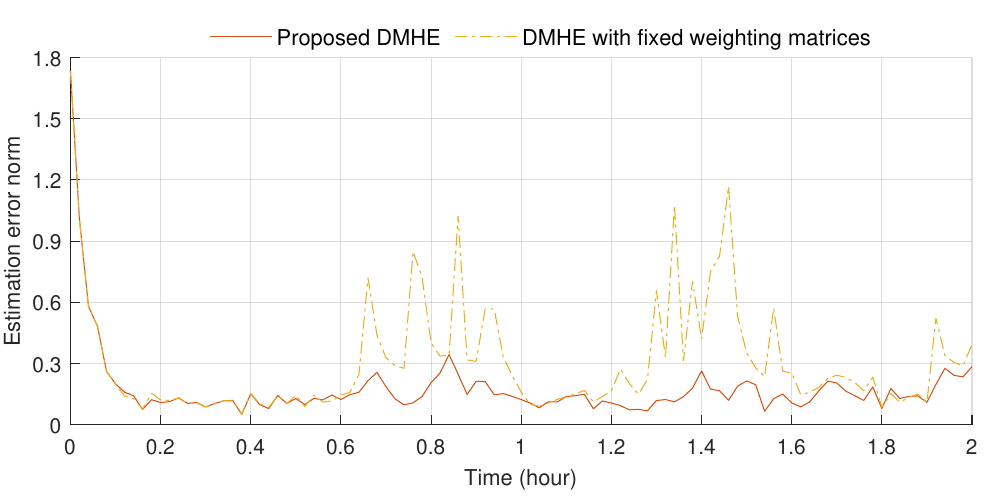}
  \caption{Trajectories of the estimation error norm of the proposed proposed DMHE algorithm and the DMHE-3 approach.}\label{compare_error}
\end{figure}

{\color{black}To further illustrate the superiority of employing a recursive approach to approximate and update the arrival costs at each sampling instant,
we compare it with DMHE-3, where the arrival costs of local estimators are designed
as a weighted squared error between the estimate of state and the initial guess of state, weighted by a constant matrix throughout the simulation period. Different from DMHE-1 in the previous comparison, DMHE-3 incorporates sensor measurements from the interconnected subsystems into the objective function of the local estimator design in the same manner as our proposed DMHE approach.
In this comparison, we randomly select the weighting matrices of DMHE-3 without fine-tuning for this process. Specifically, the constant weighting matrices $P_{i}$, $Q_{i}$, and $R$ for the $i$th estimator of DMHE-3 are diagonal matrices with the main diagonal elements set to 1, 0.001, and 0.001, respectively.
These matrices for DMHE-3 also serve as the initial weighting matrices for our proposed DMHE algorithm. In contrast, the matrix $P_{i}$ will be updated following \eqref{eq:upd:arrival cost_2} and \eqref{eq:linearization} at every sampling instant when conducting state estimation using our proposed DMHE approach.

The trajectories of actual states and the estimates provided by both the proposed DMHE algorithm and DMHE-3 are shown in Figure \ref{compare}, and the corresponding trajectories of the estimation errors are presented in Figure \ref{compare_error}. The results demonstrate that the proposed DMHE approach outperforms DMHE-3 overall in terms of estimation accuracy.
It is worth mentioning that when constant weighting matrices are employed for DMHE, extensive trial and error analysis is typically needed to fine-tune the weighting matrices for good estimation performance. In contrast, the arrival cost of each local estimator is updated at each sampling instant in our proposed DMHE algorithm, which allows for less accurate initial parameters and is more favorable for implementation.
}
{\color{black}
\begin{remark}
Compared with the iterative DMHE approaches in \cite{schneider2015convergence,schneider2015iterative,li2023iterative} that require iterative executions within each sampling period, our proposed DMHE method offers more efficient computation. This improvement is attributed to the update of the arrival cost at each sampling instant. Our approach employs a recursive method to provide a more accurate approximation of the arrival cost, which not only improves the accuracy of the state estimates but also reduces the computation complexity because the local estimators are only required to be executed once within each sampling period.
\end{remark}}

{\color{black}
\begin{remark}
One of the important tuning parameters that affect the trade-off between the estimation accuracy and the computational complexity is the length of the estimation window for the local estimators, denoted by $N$. From an application perspective, a larger $N$ has the potential to enhance the estimation accuracy of the DMHE approach. Meanwhile, increasing $N$ also leads to increased complexity of the online optimization problems associated with the local estimators, leading to a higher computational burden.
Therefore, a good trade-off between the estimation accuracy and the computational complexity should be achieved via appropriately adjusting the window length N.
\end{remark}}

\section{Concluding Remarks}
We addressed a partition-based distributed state estimation problem for general nonlinear systems. A recursive approach was introduced to approximate the arrival cost for each MHE-based estimator of the DMHE scheme. A partition-based distributed full-information estimation formulation was employed to derive an analytical expression for the arrival costs of local estimators of the DMHE algorithm in the linear unconstrained context. Based on the derived arrival cost for each local estimator, the proposed DMHE estimator for constrained linear systems was proposed, and the stability of the proposed DMHE scheme for linear systems was proven. Subsequently,
through successive linearization of nonlinear subsystem models, the arrival cost design for linear unconstrained systems was extended to the nonlinear context. Accordingly, we proposed a partition-based DMHE algorithm for constrained nonlinear processes.
The proposed DMHE method was applied to a simulated chemical process, and the results confirmed its superiority and efficacy.

{\color{black}In the future research, we will investigate the stability of the DMHE approach for nonlinear systems and the development of a robust DMHE algorithm for systems with uncertain model parameters.}

\section*{Acknowledgment}
This research is supported by the Ministry of Education, Singapore, under its Academic Research Fund Tier 1 (RG63/22), and Nanyang Technological University, Singapore (Start-Up Grant).


\begin{thebibliography}{2}


\bibitem{christofides2013distributed}
P.~D.~Christofides, R.~Scattolini, D.~M.~de~la Pe{\~n}a, J.~Liu.
\newblock Distributed model predictive control: A tutorial review and future
  research directions.
\newblock \emph{Computers \& Chemical Engineering}, 51:21--41, 2013.

\bibitem{battistelli2016stability}
G.~Battistelli, L.~Chisci.
\newblock Stability of consensus extended Kalman filter for distributed state
  estimation.
\newblock \emph{Automatica}, 68:169--178, 2016.

\bibitem{yin2018forming}
X.~Yin, J.~Zeng, J.~Liu.
\newblock Forming distributed state estimation network from decentralized estimators.
\newblock \emph{IEEE Transactions on Control Systems Technology}, 27(6):2430--2443, 2018.

\bibitem{daoutidis2019decomposition}
P.~Daoutidis, W.~Tang, A.~Allman.
\newblock Decomposition of control and optimization problems by network structure: Concepts, methods, and inspirations from biology.
\newblock \emph{AIChE Journal}, 65(10), 2019.

\bibitem{tang2018optimal}
W.~Tang, A.~Allman, D.~B.~Pourkargar, P.~Daoutidis.
\newblock Optimal decomposition for distributed optimization in nonlinear model predictive control through community detection.
\newblock \emph{Computers \& Chemical Engineering}, 111:43--54, 2018.

\bibitem{chen2020machine}
S.~Chen, Z.~Wu, D.~Rincon, P.~D.~Christofides.
\newblock Machine learning-based distributed model predictive control of nonlinear processes.
\newblock \emph{AIChE Journal}, 66(11):e17013, 2020.



\bibitem{farina2009moving}
M.~Farina,  G.~Ferrari-Trecate, R.~Scattolini.
\newblock Moving-horizon partition-based state estimation of large-scale systems.
\newblock \emph{Automatica}, 46(5):910--918, 2010.

\bibitem{schneider2015convergence}
R.~Schneider, W.~Marquardt.
\newblock Convergence and stability of a constrained partition-based moving
  horizon estimator.
\newblock \emph{IEEE Transactions on Automatic Control}, 61(5):1316--1321, 2015.




\bibitem{schneider2017solution}
R.~Schneider.
\newblock A solution for the partitioning problem in partition-based moving-horizon estimation.
\newblock \emph{IEEE Transactions on Automatic Control}, 62(6):3076--3082, 2017.

\bibitem{schneider2015iterative}
R.~Schneider, R.~Hannemann-Tam{\'a}s, W.~Marquardt.
\newblock An iterative partition-based moving horizon estimator with coupled inequality constraints.
\newblock \emph{Automatica}, 61:302--307, 2015.



\bibitem{yin2022event}
X.~Yin, B.~Huang.
\newblock Event-triggered distributed moving horizon state estimation of linear
  systems.
\newblock \emph{IEEE Transactions on Systems, Man, and Cybernetics: Systems},
  52(10):6439--6451, 2022.

\bibitem{li2023iterative}
X.~Li, S.~Bo, Y.~Qin, X.~Yin.
\newblock Iterative distributed moving horizon estimation of linear systems
  with penalties on both system disturbances and noise.
\newblock \emph{Chemical Engineering Research and Design}, 194:878--893, 2023.

\bibitem{yin2017distributed}
X.~Yin, J.~Liu.
\newblock Distributed moving horizon state estimation of two-time-scale
  nonlinear systems.
\newblock \emph{Automatica}, 79:152--161, 2017.

\bibitem{farina2011moving}
M.~Farina,  G.~Ferrari-Trecate, C.~Romani, R.~Scattolini.
\newblock Moving horizon estimation for distributed nonlinear systems with application to cascade river reaches.
\newblock \emph{Journal of Process Control}, 21(5):767--774, 2011.




\bibitem{zhang2013distributed}
J.~Zhang, J.~Liu.
\newblock Distributed moving horizon state estimation for nonlinear systems
  with bounded uncertainties.
\newblock \emph{Journal of Process Control}, 23(9):1281--1295, 2013.

\bibitem{zeng2015distributed}
J.~Zeng, J.~Liu.
\newblock Distributed moving horizon state estimation: Simultaneously handling
  communication delays and data losses.
\newblock \emph{Systems \& Control Letters}, 75:56--68, 2015.

\bibitem{qu2009computation}
C.~C.~Qu, J.~Hahn.
\newblock Computation of arrival cost for moving horizon estimation via
  unscented Kalman filtering.
\newblock \emph{Journal of Process Control}, 19(2):358--363, 2009.

\bibitem{lopez2011constrained}
R.~L{\'o}pez-Negrete,  S.~C.~Patwardhan, L.~T.~Biegler.
\newblock Constrained particle filter approach to approximate the arrival cost
  in moving horizon estimation.
\newblock \emph{Journal of Process Control}, 21(6):909--919, 2011.

\bibitem{rao2001constrained}
C.~V.~Rao, J.~B.~Rawlings, J.~H.~Lee.
\newblock Constrained linear state estimation -- a moving horizon approach.
\newblock \emph{Automatica}, 37(10):1619--1628, 2001.

\bibitem{gharbi2019proximity}
M.~Gharbi, C.~Ebenbauer.
\newblock Proximity moving horizon estimation for linear time-varying systems and a Bayesian filtering view.
\newblock \emph{IEEE Conference on Decision and Control}, 3208--3213, Nice, France, 2019.


\bibitem{rao2002constrained}
C.~V.~Rao, J.~B.~Rawlings.
\newblock Constrained process monitoring: Moving-horizon approach.
\newblock \emph{AIChE Journal}, 48(1):97--109, 2002.


\bibitem{gharbi2020proximity}
M.~Gharbi, F.~Bayer, C.~Ebenbauer.
\newblock Proximity moving horizon estimation for discrete-time nonlinear systems.
\newblock \emph{IEEE Control Systems Letters}, 5(6):2090--2095, 2020.





\bibitem{li2023data}
X.~Li, S.~Bo, X.~Zhang, Y.~Qin, X.~Yin.
\newblock Data‐driven parallel Koopman subsystem modeling and distributed moving horizon state estimation for large‐scale nonlinear processes.
\newblock \emph{AIChE Journal}, 70(3):e18326, 2024.

\bibitem{li2023partition}
X.~Li, A.~W.~K.~Law, X.~Yin.
\newblock Partition-based distributed extended Kalman filter for large-scale nonlinear processes with application to chemical and wastewater treatment processes.
\newblock \emph{AIChE Journal}, 69(12):e18229, 2023.


\bibitem{NMPC}
X.~Li, X.~Yin.
\newblock A recursive approach to approximate arrival costs in distributed moving horizon estimation.
\newblock {\em IFAC Conference on Nonlinear Model Predictive Control}, Accepted.

\bibitem{findeisen1997moving}
P.~K.~Findeisen.
\newblock \emph{Moving horizon state estimation of discrete time systems}.
\newblock Ph.D. thesis, University of Wisconsin--Madison, 1997.

\bibitem{knufer2023nonlinear}
S.~Kn{\"u}fer, M.~A.~M{\"u}ller.
\newblock Nonlinear full information and moving horizon estimation: Robust global asymptotic stability.
\newblock \emph{Automatica}, 150:110603, 2023.

\bibitem{rawlings2012postface}
J.~Rawlings, D.~Mayne.
\newblock Postface to model predictive control: Theory and design.
\newblock \emph{Nob Hill Pub}, 5:155--158, 2012.

\bibitem{henderson1981deriving}
H.~Henderson, S.~Searle.
\newblock On deriving the inverse of a sum of matrices.
\newblock \emph{SIAM Review}, 23(1):53--60, 1981.

\bibitem{farina2024moving}
M.~Farina,  G.~Ferrari-Trecate, R.~Scattolini.
\newblock Moving horizon partition-based state estimation of large-scale systems--Revised version.
\newblock \emph{arXiv:2401.17933}, 2024.






\end{thebibliography}
\end{document}